\documentclass[12pt]{article}
\usepackage{amssymb}
\usepackage{times} 
\usepackage{natbib}
\usepackage{graphicx}
\usepackage{setspace}
\newcommand{\blind}{0}
\usepackage[pass,paperwidth=8.5in,paperheight=11in]{geometry}
\usepackage[english]{babel}
\usepackage[utf8x]{inputenc}
\usepackage[T1]{fontenc}
\usepackage{authblk, algorithm, algorithmic}
\usepackage{amsfonts}
\usepackage{amsmath}
\usepackage{amsthm}
\usepackage{mathtools}
\usepackage{graphicx}
\usepackage{subcaption}
\usepackage[colorinlistoftodos]{todonotes}
\usepackage[colorlinks=true, allcolors=black]{hyperref}
\usepackage{bm}
\usepackage{textcomp}
\usepackage{subcaption}
\usepackage{placeins,bbm}

\def\T{{ \mathrm{\scriptscriptstyle T} }}
\newtheorem{theorem}{Theorem}[section]

\newcommand{\norm}[1]{\Vert#1\Vert}

\newcommand{\bw}{\mathbf{w}}
\newcommand{\balpha} {\mbox{\boldmath $\alpha$}}

\newcommand{\argmin}{\operatornamewithlimits{argmin}}

\newcommand{\bs} {\mathbf{s}}

\newcommand{\bX} {{\mathbf X}}
\newcommand{\bU} {{\mathbf U}}

\newcommand{\bb} {{\mathbf b}}
\newcommand{\bz} {{\mathbf z}}

\newcommand{\by} {{\bf y}}

\newcommand{\cG}{{\cal G}}
\newcommand{\cW}{\mathcal{W}}

\newlength\myindent
\newcommand{\red}[1]{\textcolor{red}{#1}}

\begin{document}
\def\spacingset#1{\renewcommand{\baselinestretch}%
{#1}\small\normalsize} \spacingset{1}


\if0\blind
{
  \title{\bf Generative Multi-purpose Sampler \\ for Weighted M-estimation}
  \author{
    Minsuk Shin
    
    Department of Statistics, University of South Carolina,\\
    Shijie Wang \\
    Department of Statistics, University of South Carolina,\\
    and \\
    Jun S. Liu\\
    Department of Statistics, Harvard University\\}
    \date{}
  \maketitle
} \fi

\if1\blind
{
  \bigskip
  \bigskip
  \bigskip
  \begin{center}
    {\LARGE\bf \bf Generative Multi-purpose Sampler \\ for Weighted M-estimation}
\end{center}
  \medskip
} \fi

\bigskip
\begin{abstract}
\setstretch{1.5}
To overcome  computational bottlenecks of various data perturbation procedures such as the bootstrap and cross validations, we propose  the {\it Generative Multi-purpose Sampler} (GMS), which directly constructs a generator function to produce solutions of weighted M-estimators from a set of given weights and tuning parameters. The GMS is implemented by a single optimization procedure without having to repeatedly evaluate the minimizers of weighted losses, and is thus capable of significantly  reducing the computational time. We demonstrate that the GMS framework enables the implementation of various statistical procedures that would be unfeasible in a conventional framework, such as iteratedbootstrap procedures and cross-validation for penalized likelihood. To construct a computationally efficient generator function, we also propose a novel form of neural network called the \emph{weight multiplicative multilayer perceptron} to achieve fast convergence. 
An \texttt{R} package  called \texttt{GMS} is provided, which runs under \texttt{Pytorch}  to implement the proposed methods and allows the user to provide a customized loss function to tailor to their own models of interest.     


\end{abstract}

\vfill

\newpage
\spacingset{1.5} 

\newpage

\spacingset{1.5} 
\section{Introduction}\label{sec:GMS}
Consider a canonical setting in which $\by = \{y_1,\dots,y_n\}$ are i.i.d. observations following a statistical model with the parameter of interest  denoted by  $\theta\in\Theta\subset\mathbb{R}^p$.
In some instances such as regression analysis, one may also include predictors or covariate variables for each observation.
An efficient estimator of $\theta$ can often be found by solving the following (penalized) optimization problem: $\hat{\theta}=\argmin_{\theta} L_{\by}(\theta)$, where $L_{\by}(\theta)\equiv \frac{1}{n}\sum_{i=1}^n \ell_\eta(\theta;y_i)$ with $\ell_\eta( \cdot )$ being a suitable loss function with an auxiliary parameter $\eta$.
The resulting $\hat{\theta}$ is often referred to as an {\it M-estimator} \citep{huber1992robust}. For example,  the maximum likelihood estimator (MLE) is a special  M-estimator with the loss function being set as the negative log-likelihood function.

To assess the variability of the M-estimator $\hat{\theta}$, we study behaviors of the following { tunable weighted M-estimators} as inspired by the bootstrap methods \citep{efron1979bootstrap}: \vspace{-0.1cm}
\begin{equation}\label{eq:loss}
\hat{\theta}_{\bw, \lambda,\eta}=\argmin_{\theta} \left[
\frac{1}{n}\sum_{i=1}^n w_i\ell_\eta(\theta;y_i)+\lambda u(\theta)\ \right]
\stackrel{\Delta}{=} \argmin_{\theta}
L_{\by}(\theta;\bw,\lambda,\eta), 
\vspace{-0.1cm}
\end{equation} 
where  $\eta\in\mathbb{R}^+$ is an auxiliary parameter of the loss, $u(\cdot)$ is a penalty function on the parameter with a tuning parameter $\lambda$ that can be set to zero for non-penalized settings, and 
 $\bw=(w_1,\dots,w_n)^\top \in \cW$ 
is a vector of weights following distribution $\pi(\bw)$. 
 The auxiliary parameter $\eta$ tunes the loss function. For example,   in quantile regression models, $\eta\in(0,1)$ represents the quantile level and the loss function takes the form 
 $\ell_\eta(\theta; y_i, X_i) = \rho_\eta(y_i - X_i^\top\theta), \ \ \ \mbox{where } \ \rho_\eta(t) = t(\eta- I(t<0))$. When the loss function has no auxiliary parameter, we simply denote the loss and the resulting estimator  by $\ell(\theta;y_i)$ and $\hat\theta_{\bw,\lambda}$, respectively. 

The formulation of \eqref{eq:loss} applies to a wide range of statistical procedures. For example, the classical  bootstrap procedure of \cite{efron1979bootstrap} corresponds to  $\bw\sim\text{Multinom}(n, \mathbbm{1}_n/n)$, where ${\mathbbm{1}_n}$ is a $n$-dimensional vector of one, and $u(\theta)=0$.
Random-weight bootstrap procedures 
can be formulated  by imposing a general distribution on $\bw$ that has a mean of one, finite variance, and sum to $n$. Its theoretical properties such as consistency  have been studied \citep{praestgaard1993exchangeably,cheng2010bootstrap,barbe2012weighted}. A special and most well-known form of the random-weight bootstrap is to set $\bw\sim n\times \text{Dirichlet}(n;\mathbbm{1}_n)$ as in the {\it Bayesian Bootstrap} \citep{rubin1981bayesian} and  {\it Weighted Likelihood Bootstrap} \citep{newton1994approximate}. Theoretical investigations and improvements of  the bootstrap methods have been considered in a large body of literature \citep{chatterjee2005generalized,mccarthy2018calibrated,hall1988bootstrap,efron1987better,hahn1995bootstrapping,kleiner2014scalable}. 

Iterated bootstrap procedures
are often employed to reduce 
the bias associated with a statistical inference procedure and/or improve the coverage precision of confidence intervals \citep{hall1988bootstrap}. A most frequently cited procedure is the double bootstrap, which  first bootstraps and infers the parameter or prediction, and then estimates the bias of each bootstrapped solution 
via a second-level bootstrap. In \eqref{eq:loss}, the double bootstrap procedures can be represented by setting a hierarchical weight distribution such that ${\bs}=\{s_1,\dots,s_n\}\sim \text{Multinom}(n,
\mathbbm{1}_n/n)$ and $\bw\mid {\bs}\sim \text{Multinom}(n, \bs/n)$.
These iterated bootstrap methods can be shown to provide more accurate confidence coverage  (i.e., the second  or higher-order accuracy)  compared with single bootstraps and asymptotic approximations \citep{martin1992double, mccarthy2018calibrated,hall2013bootstrap,lee1999effect,lee1995asymptotic}. However, iterative bootstraps are computationally very expensive and are rarely used in practice  when the data are of moderate to large sizes.

The tunable weighted M-estimation in \eqref{eq:loss} can also represent $K$-fold cross-validation. 
For pre-selected folds, such as a group of sample indices $I_1,\dots,I_K$, we set $w_i=0$ for $i$ in the fold of interest, say $I_1$, and set $w_i=1$ in all other folds. 
This means that the observations in $I_1$ will be ignored during training, rendering $I_1$ to be test samples. If $u(\cdot)=\norm{\cdot}_1$, the evaluated $\hat\theta_{\bw,\lambda}$ is equivalent to the LASSO estimator \citep{tibshirani1996regression}, based on a tuning parameter $\lambda$, trained without using the samples in the considered fold $I_1$, resulting in a cross-validated LASSO. The computational burden of the cross-validation linearly increases with the fold size $K$ and the candidate set size of the tuning parameter, and a typical amount is at least a few hundreds of repetitive computations. 


While aforementioned weighted M-estimation procedures  are widely used  in statistics and science, the computational bottleneck caused by their repetitive nature poses significant practical difficulties. 
To alleviate these computational difficulties, we propose a computational strategy based on a  neural network-based generative process, called the \emph{Generative Multi-purpose Sampler} (GMS) (with the \emph{Generative Bootstrap Sampler} (GBS) as a special case for bootstrap).
Instead of repeating the same optimization process for various combinations of weights $\bw$'s and parameters $\lambda$'s and $\eta$'s, the GMS constructs a generator function that takes $(\bw,\lambda,\eta)$ as input and returns the corresponding weighted M-estimator $\hat\theta_{\bw,\lambda,\eta}$. In addition to taking advantage of the high representation power of neural networks, a key idea  for the GMS to achieve the desired computational efficiency gain is to minimize an integrative loss in the training of GMS, which optimizes both the M-estimation and the parameters employed by the GMS simultaneously.

The rest of the article is organized as follows. Section~\ref{sec:GMS_sub} introduces the general GMS framework and uses a toy example to explain its potential gains. Section~\ref{sec:bootex} details its specialization for the bootstrap, namely the {\it generated bootstrap sampler} (GBS). Section~\ref{sec:CV} discusses the training of GMS for cross-validation with Lasso and quantile regression. Section~\ref{sec:computation} provides details on the neural network structures and detailed computational aspects of  GMS. Section~\ref{sec:conclusion} concludes with a brief discussion.

\section{Generative Multi-purpose Sampler}
\label{sec:GMS_sub}


\subsection{The basic formulation}
We view the weighted M-estimator $\hat{\theta}_{\bw,\lambda,\eta}$ as a function of the weight $\bw$, the tuning parameter $\lambda$, and the auxiliary parameter $\eta$, i.e., $G(\bw,\lambda, \eta)$, and attempt to approximate it by a member in a suitable family of functions $\cG =\{  G_{\phi}:\mathbb{R}^{n+2} \mapsto\mathbb{R}^p, \phi\in\Phi  \}$, where $\Phi$ is the space of parameters that characterize a function in the family. By doing so,
we  turn the unrestricted optimization problem in (\ref{eq:loss}) into a  restricted optimization problem in the functional space, i.e.,  finding a proper parameter of  the generator function  such that, for all $\bw\in\mathcal{W}$, $\lambda\in \mathbb{R}^+$, and $\eta\in \mathbb{R}^+$, 
\vspace{-0.4cm}
\begin{equation}\label{eq:GMS-loss}
\widehat \phi =  
\argmin_{\phi\in \Phi  }
L_\by(G_\phi(\bw,\lambda,\eta);\bw,\lambda,\eta),
\vspace{-0.4cm}
\end{equation}
A slightly less ambitious, but more robust, 
formulation is to solve \vspace{-0.4cm}
\begin{equation}\label{eq:GMS}
\widehat \phi = \argmin_{\phi\in \Phi } \mathbb{E}_{\bw,\lambda,\eta}\left[ L_\by(G_\phi(\bw,\lambda,\eta);\bw,\lambda,\eta)\right],
\vspace{-0.4cm}
\end{equation}
where $\mathbb{E}_{\bw,\lambda,\eta}( \cdot )$ is taken with respect to a proper distribution of $(\bw,\lambda,\eta)$ defined on $\cW \times \mathbb{R}^+\times \mathbb{R}^+$. We name this generative framework in \eqref{eq:GMS} as the GMS.
For non-penalized settings without the auxiliary parameter $\eta$, we simply denote the generator function by $G(\bw)$. We also use the notation $\widehat G = G_{\widehat\phi}$.  
The weight distribution for  Efron's nonparametric bootstrap is simply $\bw \sim \text{Multinom}(n,\mathbbm{1}_n/n)$. For the Bayesian bootstrap \citep{rubin1981bayesian},  $\bw/n\sim \text{Dirichlet}(n,\mathbbm{1}_n)$. The distributions of $\lambda$ and $\eta$ can simply be the uniform distribution on candidate sets of $\lambda$'s and $\eta$'s  chosen by the researcher. Another reasonable distribution of $\lambda$ and $\eta$ is to add random noises to a discrete set of candidate values (see Section \ref{sec:comp} for details). 


Suppose that $\hat\phi$ is the solution of \eqref{eq:GMS} for a sufficiently large family $\cG$ and a proper distribution on $\{\bw,\lambda,\eta\}$, $\mathbb{P}_{\bw,\lambda,\eta}$, supported on $\cW \times \mathbb{R}^+\times \mathbb{R}^+$. 
If the solution $\hat \theta_{\bw,\lambda,\eta}$  of \eqref{eq:loss} is unique for any given $(\bw,\lambda,\eta)$ in the support,
then $G_{\hat{\phi}}(\bw,\lambda,\eta)$ should be very close to $\hat \theta_{\bw,\lambda,\eta}$ almost surely in $\mathbb{P}_{\bw,\lambda,\eta}$. It is easy to see this point by contradiction -- if not, then there exist $\epsilon>0$ and  a subset ${S}^* \subset \cW \times \mathbb{R}^+\times \mathbb{R}^+$ 
  such that 
$\mathbb{P}_{\bw,\lambda,\eta}(S^\ast)>0$ and  $G_{\hat{\phi}}(\bw,\lambda,\eta)\leq  \hat \theta_{\bw,\lambda,\eta}-\epsilon$ on $S^*$. Thus, we can find another function that differs from  $G_{\hat{\phi}}$ only on ${S}^*$ and achieves a smaller   value in \eqref{eq:GMS}. 

A main takeaway from this argument  is that optimizing the integrative loss over the space of $(\bw,\lambda,\eta)$ instead of the individual loss is appropriate for training.
To benefit from this formulation, we must choose an appropriate family $\cG$ of functions $G_\phi$ and a suitable distribution $\mathbb{P}_{\bw,\lambda,\eta}$ to cover the hyperparameter space of interest. 
As demonstrated by our empirical studies on a wide range of problems, 
restricting $\cG$ to be a class of neural networks and choosing  a reasonable distribution $\mathbb{P}_{\bw,\lambda,\eta}$ appears to work well (see details in Section \ref{sec:comp}).

As shown in  \cite{cybenko1989approximation} and \cite{ lu2017expressive},  {\it Multi-Layer Perceptrons} (MLP), or equivalently, {\it Feed-forward Neural Networks} (FNNs), are theoretically capable of   approximating any Lebesgue integrable function when the numbers of neurons and layers are sufficiently large.  Also, recent successful applications of deep neural networks in a variety of  data-rich fields provide compelling evidence supporting the use of over-parameterized MLPs and other types of neural networks for approximating extremely complicated functions \citep{goodfellow2014generative,arjovsky2017wasserstein}. 
To train a neural network to achieve the task in \eqref{eq:GMS}, we employ a  {backpropagation} algorithm \citep{rumelhart1986learning} along with {\it Stochastic Gradient Descent} (SGD)  and its variants. 
More details are given in Section \ref{sec:NN}.

\subsection{Intuitions for potential gains}
Imagine that we have independent weight vectors $(\bw^{(1)},\lambda^{(1)},\eta^{(1)}),\ldots, (\bw^{(M)},\lambda^{(M)},\eta^{(M)})$ from $\mathbb{P}_{\bw,\lambda,\eta}$, we can approximate the expectation in (\ref{eq:GMS}) by \vspace{-0.4cm}
\begin{equation}\label{eq:MC_Approx}
\mathbb{E}_{\bw,\lambda,\eta}\left[ L_\by(G(\bw,\lambda,\eta);\bw,\lambda,\eta)\right]\approx 
    \frac{1}{M} \sum_{m=1}^M L_\by(G(\bw^{(m)},\lambda^{(m)});\bw^{(m)},\lambda^{(m)},\eta^{(m)}).
    \vspace{-0.4cm}
\end{equation}
$M$ do not needs to be very large ($M$=100, say) since a small number of samples of $(\bw,\lambda,\eta)$ can be generated continuously within the iterative SGD  algorithm to aid the fitting: after updating the FNN parameter $\phi$ with SGD based on (\ref{eq:MC_Approx}), we use the newly created samples to evaluate the fit and to provide refreshed gradient.
Thus, the two optimization tasks, i.e., minimizing the loss function $L_\by$ and finding optimal $\phi$ for the generator $G(\cdot)$, co-evolve and help each other.

If we were to cast the task of training a generator in a classical machine learning framework, we would have to first obtain a set of training samples, $\{(\bw^{(b)},\lambda^{(b)}, \hat{\theta}^{(b)})\}_{b=1}^B$, where $\hat\theta^{(b)}=\hat\theta_{\bw^{(b)},\lambda^{(b)}}$,  by evaluating $B$ optimizations in \eqref{eq:loss} with $(\bw^{(b)},\lambda^{(b)})$ for $b=1,\dots,B$ (ignoring $\eta$ for simplicity in this case). Then, one may try to learn a function $g$ by minimizing \vspace{-0.4cm}
\begin{equation}\label{eq:classicML}
\hat g = \argmin_g\sum_{b=1}^B\norm{\hat\theta^{(b)} - g(\bw^{(b)},\lambda^{(b)})}^2,
\vspace{-0.1cm}
\end{equation}
under the $l_2$-distance $\norm{\cdot}$.
However, this squared-loss  only measures the distance between the fitted generator $\hat{g}( \bw,\lambda)$ and its training true value $\hat{\theta}_{\bw,\lambda}$. As a result, it cannot inform us how to improve the fitting of the original statistical loss in \eqref{eq:loss} other than a simple interpolation. Thus, the function trained in this manner tends to be inaccurate if $B$ is small,  or may be prohibitively expensive in computation if we must rely on a large $B$, in which case
computational advantages of the generative process would be non-existing or limited.

Training the generator function $G$ in conjunction with minimizing the loss function via the GMS  formulation \eqref{eq:GMS} is significantly more efficient. The classical loss \eqref{eq:classicML} fits only on the  training data with a limited size,  $\{(\bw^{(b)},\lambda^{(b)}, \hat{\theta}^{(b)})\}_{b=1}^B$, resulting in an over-fitting issue. The GMS, on the other hand, is trained using the weights and tuning parameters generated from a predefined distribution without requiring additional optimizations for \eqref{eq:loss}, and  generating $\bw$ and $\lambda$  is nearly cost-less. As a result,  the GMS training procedure not only seeks the minimizer of $L_\by(\theta;\bw,\lambda)$, but also allows for the use of an almost infinite number of training weights and tuning parameters during the training step, thereby avoiding  over-fitting.    
\subsection{Illustration with a simple example}\label{sec:simple_example}
A novel aspect of our formulation is
represented by the minimization of the integrative loss (\ref{eq:GMS}), which combines the individual optimization step required by each classical replication  with the approximation of the functional form $G$.
Let us consider the bootstrap procedure for a toy  linear regression example with data $(y_i, X_i),$ $i=1,\ldots,n$, and the loss function $\ell(\theta;y_i,X_i) = (y_i-X_i^\top\theta)^2$ and $\lambda=0$. 
For this problem, we can obtain the closed-form  solution of the optimization problem for each bootstrapped sample: $G_0(\bw) = (\bX^\top W \bX)^{-1} \bX^\top W \by$, where $\by=(y_1,\ldots,y_n)^\top$, $\bX=(X_1^{\top},\ldots, X_n^{\top})^\top$, and $W=\text{diag}(\bw)$.
Thus, a bootstrap procedure would follow  simple steps: for $b=1,\ldots, B$,
generate $\bw^{(b)}=(w_1^{(b)}, \dots, w_n^{(b)})\sim \text{Multinom}(n,\mathbbm{1}_n/n)$ or $n\times \text{Dirichlet}(n,\mathbbm{1}_n)$, 
and
then for each $\bw^{(b)}$, plug in the formula to get   $\hat{\theta}^{(b)}=G_0(\bw^{(b)})$. However, if one does not have the closed-form formula but has to solve numerically the minimization problem  of \eqref{eq:loss} for every generated $\bw^{(b)}$, the bootstrap procedure can be prohibitively demanding in computation. Thus, our GMS formulation via \eqref{eq:GMS} can be thought of as an automatic way to find a highly accurate approximation to  the closed-form solution (in the form of a neural network) of  the minimization problem of \eqref{eq:loss}.
Once this solution $\widehat G$ is found, one can easily generate bootstrap estimators with almost no computational cost.

 For a case of $n=100$ and $p  = 10$, we set  the true coefficient $\theta=\{1,0,\dots,0\}$ and the regression variance one. The predictors are independently generated from $N(0,I_p)$. Even though this example is simple, constructing the generator function is non-trivial, because the generator function's domain dimension is   100, and the dimension of its codomain is 10. Training a 100-dimensional function with 10-dimensional codomain  requires a large number of training samples in the framework of \eqref{eq:classicML}. This fact is reiterated in Figure \ref{fig:classicML}. We generate a data set and evaluate random weight bootstrap estimators with $\bw\sim n\times\text{Dirichlet}(n;\mathbbm{1}_n)$, and  then numerically evaluate the average loss of \eqref{eq:loss} on various weights from the trained generator for the classical machine learning approach with $B=500$ and $B=5,000$, as well as the GMS. 
 We initialize the optimization in different five points for each procedure. 
 
  We consider two performance measures for this  example: the training loss specified in \eqref{eq:classicML} and the \emph{integrative prediction loss} (IPL) that can be defined as 
  $\mathbb{E}_\bw\norm{\hat\theta_{\bw} - g(\bw)}^2$.
  The IPL is approximated by using $1,000,000$ Monte Carlo evaluations, and the loss values are multiplied by $n$ to adjust for the scale of $\text{Var}(\hat \theta)$. Note that the GMS trains its generator $G$ by minimizing the integrative loss \eqref{eq:GMS}, whereas the naive generator $g$ is trained using the $l_2$-loss in \eqref{eq:classicML} with  $B=500$ and $5,000$ training samples, respectively. As expected, Figure~\ref{fig:classicML}(a) shows that the training $l_2$-losses for the naive procedures are significantly lower than those for the GMS. 
  However, the IPLs of the considered methods behave quite differently. The naive minimizers (for the cases with $B=500$ and $5,000$)  first decrease their IPLs rapidly, but after 200 iterations their IPLs begin to increase. In contrast, the GMS  seamlessly reduces its IPL. The poor predictive performance of the naive procedure stems from the fact that the $l_2$-loss encourages the generator function $\widehat g$ to overfit the training set $\hat \theta^{(1)},\dots,\hat \theta^{(B)}$. Unlike the conventional machine learning modeling, the GMS is quite resistant to overfitting, as we can sample $\bw$'s at near-zero computational cost during the  training of the generator function. 

\begin{figure}[t]
\vspace{-1.5cm}

    \centering
    \begin{subfigure}[b]{0.42\textwidth}
        \includegraphics[width=\textwidth]{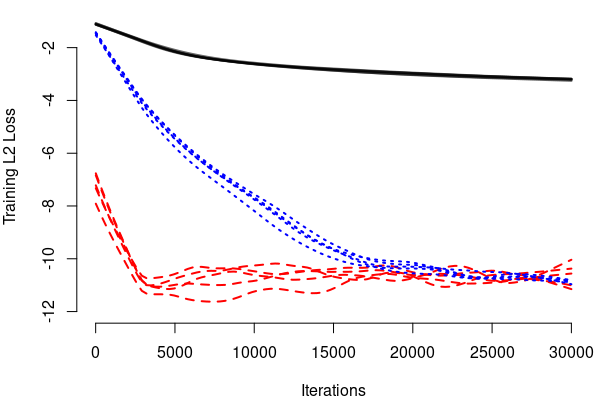}
        \vspace{-0.7cm}
        \subcaption{\small $\log (n\sum_{b=1}^B\norm{\hat\theta_{\bw^{(b)}} - g(\bw^{(b)})}^2/B)$.}
    \end{subfigure}
    \begin{subfigure}[b]{0.42\textwidth}
        \includegraphics[width=\textwidth]{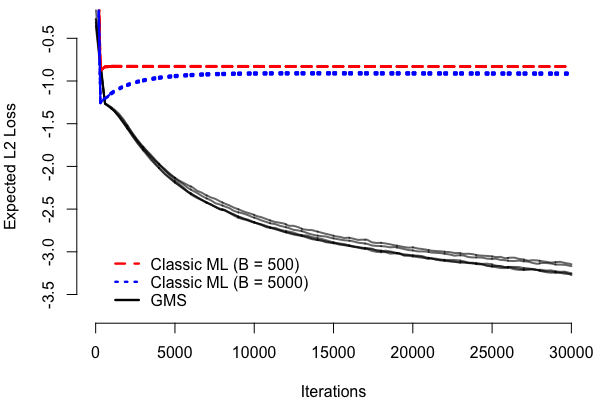}
        \vspace{-0.7cm}
        \subcaption{\small$\log (n\mathbb{E}_\bw[\norm{\hat\theta_{\bw} - g(\bw)}^2])$.}
    \end{subfigure}
    \caption{\small Trace plots of (a) the training loss,  and (b) the integrative prediction loss, in the logarithmic scale. Five lines for each optimization represent five distinct initializations; and the red dashed and blue dotted lines indicate the conventional ML with $B=500$ and $B=5000$, respectively. }
    \label{fig:classicML}
\end{figure}

 \section{Generative Bootstrap Samplers}\label{sec:bootex}
  \subsection{Bootstrap and subgroup bootstrap}\label{sec:Blockbootstrap}
   
   The simplest use of the GMS is to bootstrap M-estimators, which is a special case of  form \eqref{eq:GMS} without $\eta$ and $u(\cdot)$. The weight distribution is $ \text{Multinom}(n,\mathbbm{1}_n/n)$ (or $\ n\times\text{Dirichlet}(n,\mathbbm{1}_n)$ for the  Bayesian bootstrap). More precisely, we let $\phi$ be the parameter underlying the generator $G$ and solve the optimization problem: $\widehat \phi = \argmin_{\phi\in \Phi } \mathbb{E}_{\bw}\left[ \frac{1}{n}\sum_{i=1}^nw_i\ell(G_\phi(\bw);y_i)\right]$.
 We call this simple GMS application the \emph{Generative Bootstrap Sampler} (GBS). 
 
Despite its considerable efficiency, the GBS framework has a fundamental limitation for practical bootstrap applications: the dimension of the generator domain equals the sample size $n$. 
Even when computationally efficient neural networks are used to model the generator, the convergence  is quite slow when the input dimension is high (say, tens of thousands).
We may further encounter  technical issues such as memory shortage as well, which is particularly severe for big data.  
To address this limitation, we consider a subgroup weighting strategy, which divides the data set into subgroups and assigns equal weights to observations within each subgroup. The  subgrouping idea is primarily used for bootstrapping time series data sets, referred to as {\it block bootstrap} \citep{lahiri1999theoretical,hardle2003bootstrap}, in order to preserve the temporal association within bootstrapped samples. In contrast to the time series applications, we use subgrouping (or blocking) to reduce the number of weights, or more precisely, the domain dimension of the generator function so as to save computational costs. 

Let $[n]$ denote the index set $\{1,\ldots,n\}$ of the observations. We consider an exclusive and exhaustive partition: $I_1,\dots,I_S $ $\subset [n]$ such that $I_i \cap I_j =\emptyset, \forall i\neq j$, and  $\cup_{s=1}^S I_s= [n]$. 
Without loss of generality, we assume that the size of each $I_s$ is the same, i.e., $|I_s| = n/S$ for $s=1,\dots,S$. We define a subgroup assignment function  $h: [n]\mapsto [S]$ such that $h(i)=s$ if $i\in I_s$.
Then, for  $\{\alpha_1,\dots,\alpha_S\}^\T\sim \mathbb{P}_{\bm{\alpha}}$,  with $\mathbb{P}_{\bm{\alpha}}$ being an $S$-dimensional weight distribution, we impose the same value of weight on all elements in a subgroup as \vspace{-0.4cm}
\begin{equation}\label{eq:sub_assign}
w_i = \alpha_{h(i)} \ \ \text{ for }  \ i=1,\dots,n.
\vspace{-0.4cm}
\end{equation}
 and we denote $\bw_{\bm\alpha}=\{\alpha_{h(1)},\dots,\alpha_{h(n)}\}^\T\in\mathbb{R}^n$. As a result, it follows that  $\alpha_{h(i)}=\alpha_{h(k)}$, if $i,k\in I_s$ for some $s$. Similar to the vanilla GBS, setting 
 $\bm{\alpha}\sim \mbox{Multinomial}(S,\mathbbm{1}_S/S)$ or $\bm{\alpha}\sim S\times \mbox{Dirichlet}(S,\mathbbm{1}_S)$  result in the block-based nonparametric bootstrap and Bayesian bootstrap,
 respectively. 


As an illustration, we consider a simple  linear regression example by generating  a data set from the model with $n=1000$, $p=10$, the coefficients being a sequence of equi-spaced values between $-2$ and $2$, and $\sigma^2=2$. Each covariate is drawn i.i.d. from $N(0,1)$, and the regression variance is set to one.   The resulting domain dimension of a vanilla $G$ is 1000. Figure \ref{fig:hist_GMS}  shows individual  histograms of bootstrap distributions with varying subgroup sizes. 
 Even when the number of subgroups is  tiny ($S=5$), the obtained bootstrap distributions are acceptable,
 although the variability tends to be underestimated. As $S$ increases ($S=25$), the quality of the approximation of the subgroup bootstrap distribution improves significantly. When $S=100$, the subgroup bootstrap distributions are indistinguishable from the target ones. 
When we use 100 subgroups (10 observations in each subgroup), the input dimension is reduced to 100 from  the original $1000$ but the resulting bootstrap distributions are nearly identical to those from the standard bootstrap (see Figure~\ref{fig:block}). 
\begin{figure}[t]
    \centering
        \begin{subfigure}[b]{0.8\textwidth}
        \includegraphics[width=\textwidth]{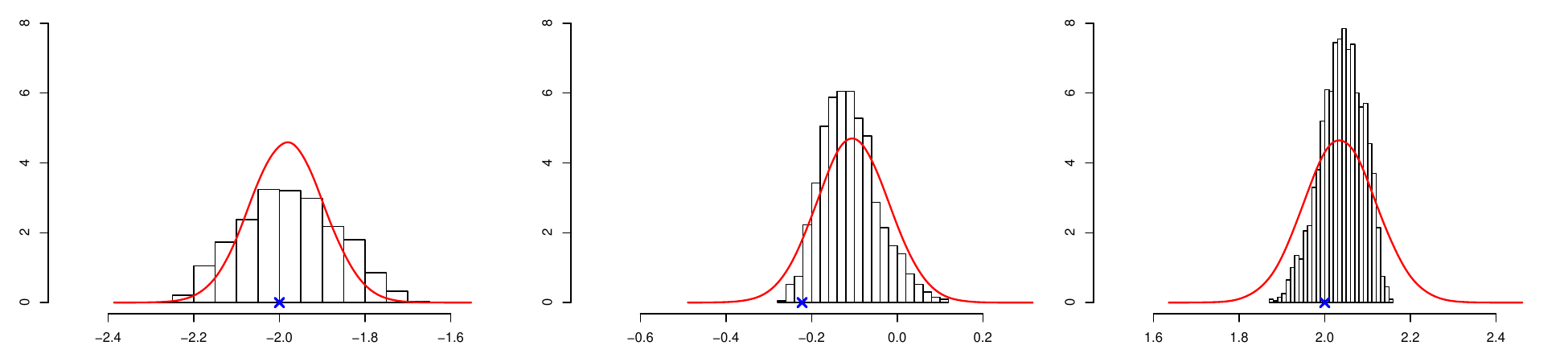}
        \subcaption{\small  $S=5$.}
    \end{subfigure}
    \begin{subfigure}[b]{0.8\textwidth}
        \includegraphics[width=\textwidth]{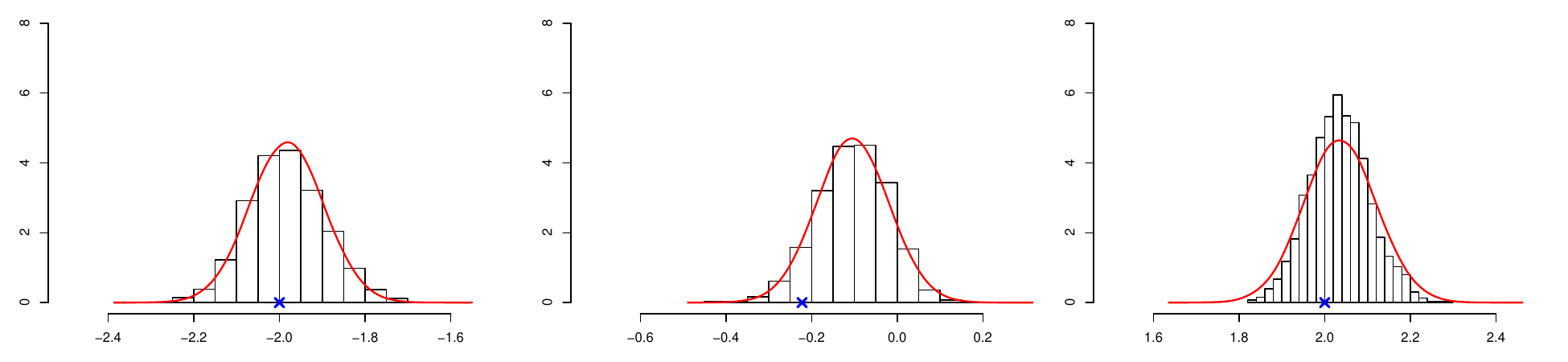}
        \subcaption{\small $S=25$.}
    \end{subfigure}
    \centering
    \begin{subfigure}[b]{0.8\textwidth}
        \includegraphics[width=\textwidth]{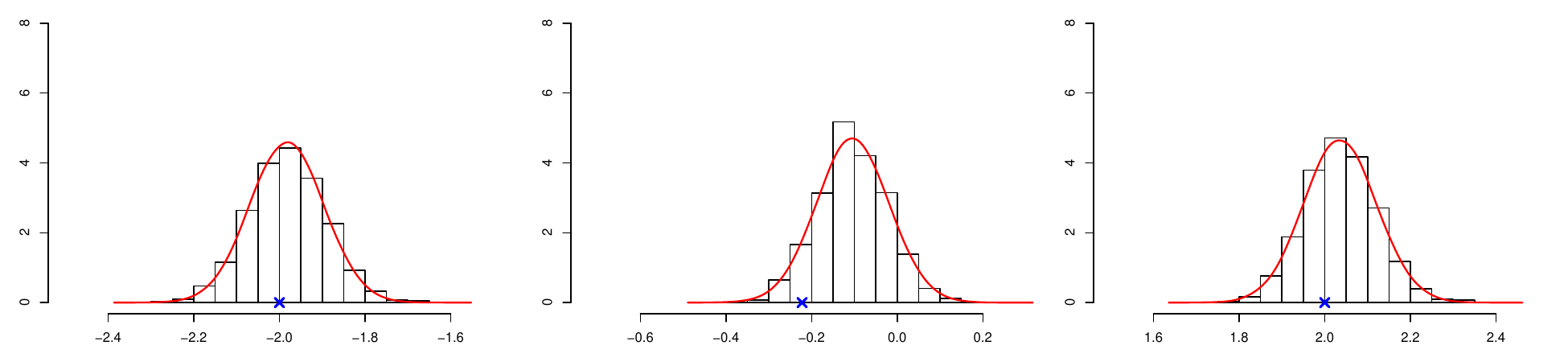}
        \subcaption{\small $S=100$.}
    \end{subfigure}
    \caption{\small Histograms of block bootstrap distributions with various $S$ for the coefficient of $X_1$ (top), $X_5$ (middle), and $X_{10}$ (bottom) for each subfigure. The red line indicates the density function of the target distribution (of the standard bootstrap). }\label{fig:hist_GMS}
\end{figure}


\noindent{\bf Remark. } 
Under some regularity conditions, one can show that the subgroup bootstrap is consistent when $S$ is of a higher order than $\sqrt{n}$ (see the Supplementary Materials for a formal proof).  

\begin{figure}[t]
    \centering
        \includegraphics[width=0.75\textwidth]{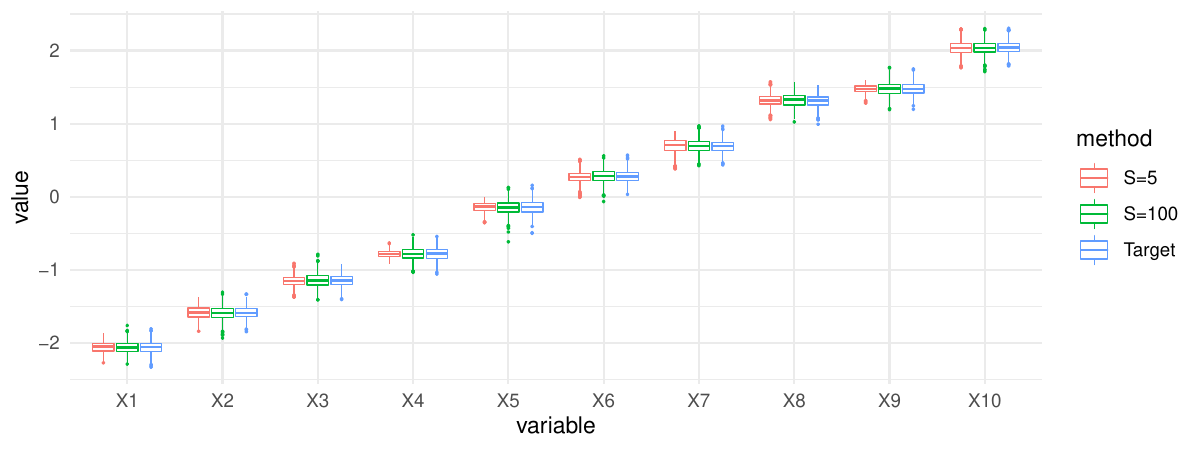}
        \vspace{-0.3cm}
    \caption{\small Comparisons of  subgroup bootstraps across different numbers of blocks. }\label{fig:block}
\end{figure}

\subsection{Iterated  bootstrap}\label{sec:double}
The iterated bootstrap method was proposed to improve the inference accuracy of the simple bootstrap method, and was shown both  theoretically and empirically  to achieve a higher-order accuracy for the coverage of the  constructed confidence intervals  and bias-corrections \citep{martin1992double, mccarthy2018calibrated,hall2013bootstrap,lee1999effect,lee1995asymptotic}. 
More precisely, an iterated bootstrap procedure
involves nested levels of data resampling. 

The double bootstrap, which is the simplest iterated bootstrap, 
first creates $B$ bootstrap samples, $\by_{b}^*$, for $b=1,\dots,B$ by resampling from the original data set, and then, for each bootstrapped sample $\by_{b}^\ast$, creates $C$ second-level bootstrap samples, $\by_{bc}^{**}$, $c=1,\ldots,C$, by resampling from   $\by_{b}^*$. For each $\by_{b}^*$ and $\by_{bc}^{**}$, we denote the corresponding estimator of $\theta$ by $\hat\theta_{b}^{*}$ and  $\hat\theta_{bc}^{**}$, respectively. By iterating this step, we can simply extend this to more iterated bootstrap cases.

 Various procedures for constructing confidence intervals using bootstrap have been proposed, such as  the \emph{percentile} method \citep{hall1992bootstrap}, the \emph{studentized} method \citep{hall1988theoretical,efron1979bootstrap}, the \emph{Bias-Corrected and accelerated} method \emph{BC$_a$}   \citep{efron1987better}, and  \emph{Approximated Bias Correction} (ABC; \cite{diciccio1992more}), etc. Even though BC$_a$ and ABC procedures enjoy the second-order accuracy (fast convergence in coverage error), a practical implementation of these procedures are not trivial since it is difficult to calculate  their acceleration factor for general models. On the other hand, the percentile procedure is only first-order correct, and the studentized procedure requires an iterated bootstrap unless an explicit form of the standard error of the bootstrap estimator is available.
To improve the quality of the constructed CI,
we consider using double bootstraps  as in the coverage calibration method \citep{hall1988bootstrap,hall1986bootstrap} and studentized CI procedure \citep{hall1988theoretical}. 
The calibrated percentile two-sided CI via  double bootstrap achieves the second-order accuracy $O(n^{-1})$, while its single bootstrap counterpart 
only attains a rate of $O(n^{-1/2})$.  However, applying the conventional double bootstrap requires undesirably intensive computation: a total of $B\times C$  evaluations of bootstrap estimators $\hat\theta_{bc}^{**}$ for $b=1,\dots,B$ and $c=1,\dots,C$. \cite{lee1999effect} showed that $B$ and $C$ should be of a higher order than $n^4$ and $n^2$ for two-sided CIs and of order $n^2$ and $n$ for one-sided CIs, respectively, so that the coverage error rate of the Monte Carlo interval is no greater than that of  the theoretical double bootstrap interval. The authors considered $B=1000$ and $C=500$ in their simulations, resulting in a total of  $500,000$ evaluations, which is an  unmanageable size  under the conventional bootstrap framework. 
\subsection{GBS for  iterated bootstrap}
Extending the GBS to  iterated bootstraps is immediate as it is a special case of \eqref{eq:GMS} with a weight distribution that has a hierarchical structure. 
For a $d$-level iterated bootstrap procedure, we may characterize its weight distribution  hierarchically: $\bw_{(1)}\sim \text{Multinom}(n, \mathbbm{1}_n/{n}),$ 
$\ldots,$ ${\bw}_{(d)}\mid {\bw}_{(d-1)}
\sim \text{Multinom}(n, \bw_{(d-1)}/{n})$.
The computational advantage of the GBS framework is particularly significant in these iterated situations.

One drawback of the standard nonparametric bootstrap is that each bootstrap sample only touches upon about $1-e^{-1}\approx63\%$ of the observations due to the nature of multinomial sampling, which appears to be somewhat wasteful. This loss is compounded and become more significant in iterated bootstraps.  
A smoothed version of these weight distributions is a hierarchy of Dirichlet distributions, which enable each $\hat{\theta}^\ast_b$ and $\hat{\theta}^\ast_{bc}$ to utilize all the observations \citep{cheng2010bootstrap,xu2020applications, praestgaard1993exchangeably}. Thus, we mainly consider  $\bw\mid{\mathbf z}\sim n\times\mbox{Dirichlet }(n,\bz)$ and $ {\mathbf z} \sim n\times\mbox{Dirichlet }(n, \mathbbm{1}_n)$. 
If a subgroup bootstrap as in  Section~\ref{sec:Blockbootstrap} is employed the subgrouped weights follow $\bw\mid{\mathbf z}\sim S\times\mbox{Dirichlet }(S,\bz)$ and $ {\mathbf z} \sim S\times\mbox{Dirichlet }(S, \mathbbm{1}_S)$. 
We train a generator function that covers both single and double bootstraps by adopting a probabilistic mixture of single and double bootstrap weights distributions; e.g., generate single or double bootstrap weights with 50\%-50\% chances.

\begin{figure}[t]
    \centering
      \begin{subfigure}[b]{0.35\textwidth}
        \includegraphics[width=\textwidth]{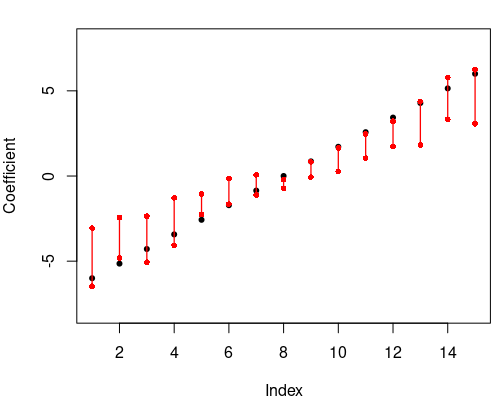}
  \end{subfigure}    
    \begin{subfigure}[b]{0.35\textwidth}
        \includegraphics[width=\textwidth]{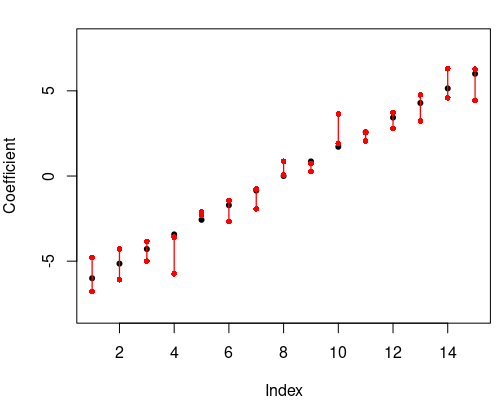}
  \end{subfigure}    
    \begin{subfigure}[b]{0.35\textwidth}
        \includegraphics[width=\textwidth]{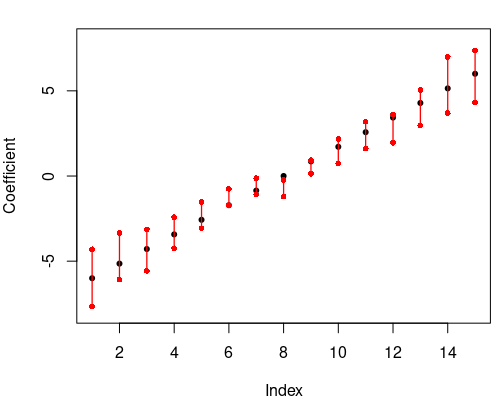}
    \end{subfigure}
\begin{subfigure}[b]{0.35\textwidth}
        \includegraphics[width=\textwidth]{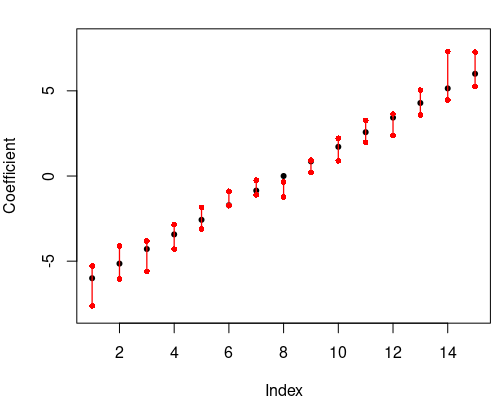}
    \end{subfigure}    
    \caption{\small 95\% CIs for the logistic regression example: The basic single bootstrap CI  (top left); the BCa CI (top right); the calibrated percentile bootstrap CI via double bootstrap (bottom left); a studentized bootstrap CI via double bootstrap (bottom right). True parameters are  marked by black dots.}\label{fig:intro_double}
\end{figure}

\subsection{An illustration: double-bootstrap for logistic regression}\label{sec:sim0}
We consider the standard logistic regression model: $y_i \sim \text{Bernoulli}\left(\frac{1}{1+\exp\{-X_i^\top \theta\}}\right)$, where $X_i\in \mathbbm{R}^p$ and  $\theta\in\mathbb{R}^p$ for $ i=1,\ldots, n$.
To apply the GBS to this model, we simply set 
$\ell (\theta;y_i) = (1-y_i)\log X_i^\T\theta + \log (1+\exp(-X_i^\T\theta)) $ in \eqref{eq:GMS}. We simulate a data set that contains  $n=400$ observations, each with $p=20$ covariates generated independently from the standard Gaussian. The true coefficient vector is set to be an equi-spaced sequence between $-6$ and $6$.         

We examine 95\% CIs constructed by various procedures, including a bias-corrected percentile CI (single bootstrap, denoted by ``basic''), a BCa CI \citep{efron1987better}, a calibrated percentile CI (double bootstrap), and a studentized CI (double bootstrap). The ``basic" CI is constructed as  $(2\hat\theta - q^*_{97.5\%}, 2\hat\theta -  q^*_{2.5\%} )$, where $q^*_{\beta}$ is the $\beta$-quantile of  the bootstrap distribution of  $\hat\theta^*$. The calibrated percentile CI is obtained as $(2\hat\theta - q^*_{\hat\alpha_U}, 2\hat\theta - q^*_{\hat\alpha_L})$, where $\hat\alpha_L$ and $\hat\alpha_U$ are calibrated coverage levels via the double bootstrap aiming at $2.5\%$ and $97.5\%$, respectively. The studentized CI is $(\hat\theta - \tilde t^*_{97.5\%}\hat s, \hat\theta - \tilde t^*_{2.5\%}\hat s)$, where $\tilde t^*_{\beta}$ is the $\beta$-quantile of the studentized bootstrap statistic, and $\hat s$ is the estimated standard error (a detailed description of these bootstrap procedures is given in Section \ref{sec:details_double} of the supplementary materials). Figure \ref{fig:intro_double} shows that, despite the fact that the single bootstrapped CI (top left) is bias-corrected, the resulting CI is strongly biased and its width is excessively wide with a low coverage  ($12/20$), and  BCa suffers from a similar issue. 
In contrast,  the two double bootstrap procedures result in better  CIs by both shortening  the widths and improving the coverage accuracy ($19/20$ coverage for both the calibrated and the studentized).

For the double bootstrapped CIs, we generate $5000$ bootstrap samples for the first-level  and 1000 for the second-level, resulting in  a total of  $5000\times 1000 = 5,000,000$ bootstrap evaluations. This  poses a significant computational challenge under the conventional framework. In comparison, once the generator function is trained (which takes less than 3 minutes for this example), the GBS produces $10,000$ bootstrap estimators in less than 0.1 second, and its computational advantage is even more significant when $n$ and $p$ are larger, as shown next.


\subsection{Scaling up towards large $n$ and  $p$} \label{sec:large_n_p}

We consider the same logistic regression model as in Section~\ref{sec:sim0} with
the true regression coefficients $\theta_j$'s   equally spaced between $-3$ and $3$.  We compare the performance of the GBS with those of the standard bootstrap, BCa, and the profile likelihood confidence interval with sample size  $n\in\{500,5000, 10000\}$ and dimension of covariates $p\in\{30,200,300\}$. This simulation is replicated independently 20 times.  
We examine properties of the 95\% CIs constructed by these bootstrap methods (i.e., the average coverage and average width, and their actual  computing time). For   standard bootstrap procedures, we consider both a parallel computing environment using 25 CPU cores (abbreviated as ``25C''), and a single-core  computation (i.e., ``1C''). 
The detailed setting is described in Section \ref{sec:comp}, and the specification of the computing server is given in the the supplementary materials. 
We use the \texttt{R} package \texttt{boot} to implement conventional bootstrap procedures.  The profile likelihood CI is based on an asymptotic approximation, and its computation is carried out  by using the \texttt{confint} function in \texttt{R}.  Due to the computational burden,  the conventional CI procedures for large sized data sets are too expensive, so we  only report the estimated computation times using two replicates. 

\begin{table}[t]

    \centering
    \resizebox{16cm}{!}{
    \begin{tabular}{c | ccr | ccr | ccr}
    \hline
    
     \hline
     
  \hline
     &  \multicolumn{3}{c}{$(n,p)=(500,30)$}  &  \multicolumn{3}{c}{$(n,p)=(5000,200)$} & \multicolumn{3}{c}{$(n,p)=(10000,300)$} \\  
     \hline
     
     \hline
    Method & Cov & Width & Time & Cov & Width& Time & Cov & Width & Time \\
\hline
  GBS1 (Basic) & 0.975 & 1.804  & 140.8 + \red{0.1} &  0.987 & 1.028 & 152.9 + \red{0.2} & 0.976 & 0.721 & 163.6 + \red{0.4}\\ 
  GBS1 (Percentile) & 0.752 & 1.804 & 140.8 + \red{0.1} & 0.217 & 1.028 &  152.9 + \red{0.2} & 0.199 &0.721  & 163.6 + \red{0.4} \\ 
  GBS2 (Student) & 0.962 & 1.548 & 140.8 + \red{15.6} & 0.960 & 0.904 &  152.9 + \red{45.0}& 0.931 & 0.651 & 163.6 + \red{63.9}\\ 
  GBS2 (Calibrated) & 0.933 & 1.463 & 140.8 + \red{15.6} & 0.954 & 0.894 & 152.9 + \red{45.0} & 0.938 & 0.661 & 163.6 + \red{63.9}\\ 
  \hline
  Basic (25C) & 1.000 & 1.899 & 8.4 & 1.000 & 1.237 &  539.6 & \texttt{NA} & \texttt{NA} & 4227.05 \\  
  Basic (1C) &  &  & 93.8 &  &  &  3833.3 &  &  & 25540.5\\  
 
  Percentile (25C) & 0.760 & 1.899 & 8.4 & 0.230 & 1.237 & 539.6  & \texttt{NA} & \texttt{NA} & 4227.05 \\ 
  BCa (25C) & 1.000 & 2.039 & 84.3 & \texttt{NA} & \texttt{NA} &  \texttt{NA} & \texttt{NA} & \texttt{NA} & \texttt{NA}\\ 
  Profile & 0.918 & 1.657 & 0.7 & 0.462 & 0.941 & 1310.8   & \texttt{NA} & \texttt{NA} & 8670.7\\ 
  \hline
    \end{tabular}
    }
    \caption{\small Results of  the simulation study for logistic regression models; ``GBS1'' and ``GBS2'' indicate single and double bootstraps implemented by the GBS, respectively; ``Cov'', ``Width'', and ``Time'' mean the averages (over 20 replicates) of the coverage, the width, and the actual computing time (seconds) of each evaluated  95\% CI, respectively; for the computation time of the GBS, the black and red numbers indicate training and generation time (including post processing time for the GBS), respectively.}
    \label{tab:sim_logit}
\end{table}

Table \ref{tab:sim_logit} compares  traditional bootstrap procedures with their GBS equivalents in various settings. The GBS procedures are comparable to their conventional counterparts (``Basic'' and ``Percentile''  in the table) in terms of the  coverage and width of the constructed CIs. When  $(n,p)=(500,30)$, the traditional bootstrap-based CIs are significantly faster to compute. However, as data size increases, the conventional bootstrap   becomes prohibitively expensive, taking more than an hour  for $(n,p)=(10000,300)$ using a parallel computation with 25C, and more than 7 hours using 1C.
Due to its heavy computational need, the BCa cannot produce meaningful results for moderately large data sets  (e.g., for $(n,p)=(5000,200)$ and $(10000,300)$). The profile likelihood procedure (``Profile''), which is based on an asymptotic approximation of the sampling distribution, is also quite expensive when data size becomes large.

For the double bootstraps, the conventional repetitive computations  take more than 2.5 hours with parallel computation using 25C for the case with $(n,p)=(500,30)$, and would have taken more than 48 days for the case with $(n,p)=(10000,300)$. As a result, the conventional double bootstrap procedures are infeasible for multiple replicates, so their results are omitted in Table \ref{tab:sim_logit}.
In contrast, the GBS training takes less than three minutes for all examined settings, while the generation and  post-processing  for the  double bootstrap take about one minute. Furthermore, 
the double-bootstrap GBS2s requires 
very little extra computational time, but achieves a significantly higher accuracy, than the single bootstrap GBS1s.
\section{Bootstrap Cross-Validation for Parameter Tuning Via GMS}\label{sec:CV}

Tuning parameter selection has been a challenging and  computationally intensive task for many statistical and machine learning algorithms since repetitive computations are often required  over a wide range of possible choices of the tuning parameter. 
We note that the GMS framework is not only applicable to bootstrap,  but can also be used to expedite the computation of  \emph{Cross-Validation} (CV) procedures. It is easy to see that for a weight $w_i=0$, the corresponding term in  the weighted M-estimation loss function \eqref{eq:loss} is  zero, 
which is equivalent to ignoring  observation $y_i$. 
More generally, we denote ${\mathbf w}_{(-I)} = \{w_1,\dots,w_n\}$ with $w_i=0$ for $i\in I$, and 
$\{w_i:i\not\in I\} \sim (n-|I|)\times\text{Dirichlet}(n-|I|;\mathbbm{1}_{n-|I|})$.
Thus, index sets $I$ and $I^c$ can be viewed as those for the test and training data, respectively. 
To train the CV generator without the bootstrapping aspect, one may employ a simpler weight distribution  than  the multinomial or Dirichlet, such as setting all the weights in a randomly selected fold to be zero, and the remaining to be one. 
Based on this setup, a simple modification of Algorithm \ref{alg:alg} (with strategies in Section \ref{sec:comp}) can be used to train the generator for the $K$-fold CV 
(more details  in the Supplementary Materials). 
Once the generator is trained, 
one can easily compute the estimated out-of-sample error across different tuning parameters by alternating zero weight for each fold.

More precisely, for $b=1,\dots,B$ and a tuning parameter $\lambda_l$ in a candidate set $\{\lambda_1,\dots,\lambda_L\}$, we set zero weights on a fold $I^*_k$ for $k=1,\dots,K$; i.e.,  $w_i^{(b,k)}=0$ for $i\in I^*_k$. 
For ${i\not\in I^*_k}$, we can set $w_i^{(b,k)}=1$ when only  CV is of interest, or let $\{w_i^{(b,k)}\}_{i\not\in I^*_k}\sim (n-|I^*_k|)\times\text{Dirichlet}(n-|I^*_k|,\mathbbm{1}_{n-|I^*_k|})$ so as to quantify uncertainty in the CV via bootstrap. The bootstrapped CV estimator without considering the test set $I^*_k$ with a tuning parameter $\lambda_l$, denoted by $\hat\theta_{(-I_{k}^*),\lambda_l}^{(b)}$, can be computed as $\widehat G(\bw^{(b,k)},\lambda_l)$. The CV loss for the $k$-th fold and $\lambda_l$ follows as $\hat e_{l}^{(b,k)} = \sum_{i\in I^*_k}\ell(\hat\theta_{(-I_{k}^*),\lambda_l}^{(b)}; y_i)/|I_k^*|$. 
After repeating this step for all the $K$ folds, we obtain the bootstrapped $K$-fold CV errors as $\bar e_{l}^{(b)}=\sum_{k=1}^K \hat e_{l}^{(b,k)}/K$.   
After obtaining $\bar e_l^{(b)}$ for $l=1,\dots,L$ and $b=1,\dots,B$, one can easily identify the bootstrap distribution of the out-of-sample loss  via the empirical distribution of $\{\bar e_l^{(b)}\}_{b=1,\dots,B}$ under $\lambda_l$, as well as confidence bands of the out-of-sample loss over $\{\lambda_1,\dots,\lambda_L\}$. 

Moreover, with $l^{(b)}=\argmin_l\{ \bar e_l^{(b)}\}$, the empirical distribution of $\{\lambda_{\text{min}}^{(b)} \stackrel{\Delta}{=} \lambda_{l^{(b)}}, \ b=1,\dots,B\}$  serves as the bootstrap distribution of the minimizer of CV errors and can naturally quantify the uncertainty of the chosen tuning parameter (an example is given in the left of Figure \ref{fig:cv}). 
For example, this bootstrap distribution  $\{\lambda_{\text{min}}^{(b)}\}$ provides us an alternative to the  \emph{ad hoc}  one-standard-error rule  commonly recommended for Lasso regression, in which one chooses the most parsimonious model whose CV error is no more than one standard deviate above that of the best model. 
In contrast, with the availability of the bootstrap distribution of $\lambda_{\text{min}}$, we may pursue a more parsimonious model by using the lower $(1-\alpha)$\% confidence bound of this  distribution as  our chosen $\lambda$.\\

\noindent{ \bf Cross-validation for LASSO and ridge regression.} 
 Two representative examples of the penalized M-estimation are ridge \citep{hoerl1970ridge} and LASSO regression models 
 \citep{tibshirani1996regression}, with the corresponding  loss function for GMS: 
 
\begin{equation}\label{eq:LASSO}
 \mathbb{E}_{{\mathbf w}, \lambda} \Big[ \frac{1}{n}\sum_{i=1}^n w_i \{y_i - X_i^\T G(\bw,\lambda)\}^2 + \lambda u(G(\bw,\lambda))  \Big], 
\end{equation}
with $u(x)=\norm{x}^2_{2}$
for the ridge regression and $u(x)=\norm{x}_{1}$ for the LASSO. 
After obtaining the trained $\widehat G$ from \eqref{eq:LASSO}, for a given input $\bw^*$ and $\lambda^*$, its output $\widehat G({\mathbf w}^*, \lambda^*)$ approximates the minimizer of $\sum_{i=1}^n w_i^*\ell(\theta; y_i)/n + \lambda^* u(\theta)$ with respect to $\theta$.
We simulated from a linear regression model with $n = 500$, $p = 50$, the true parameter $\theta_0=\{1,-2,1,0,\dots,0\}$, and $\sigma^2_0=1$. Each covariate vector $X_i$ follows iid $N(0,\Sigma)$ with $\Sigma_{kl}=1$ for $k=l$ and $\Sigma_{kl}=1/2$ for $k\neq l$. 

\begin{figure}[t]
    \centering
        \begin{subfigure}[b]{0.4\textwidth}
        \includegraphics[width=\textwidth]{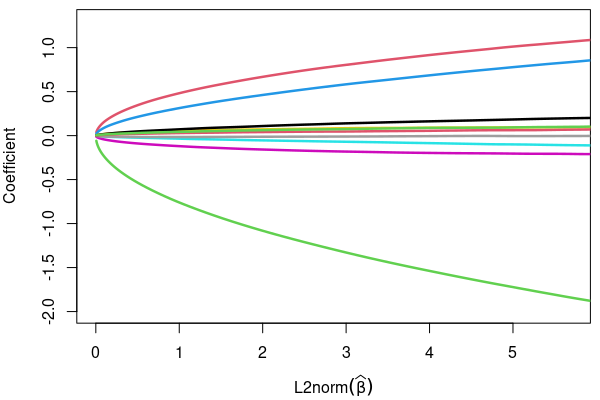}
    \end{subfigure}
        \begin{subfigure}[b]{0.4\textwidth}
        \includegraphics[width=\textwidth]{path_ridge_GMS}
    \end{subfigure}
    
\vspace{0.5cm}
    
    \begin{subfigure}[b]{0.4\textwidth}
        \includegraphics[width=\textwidth]{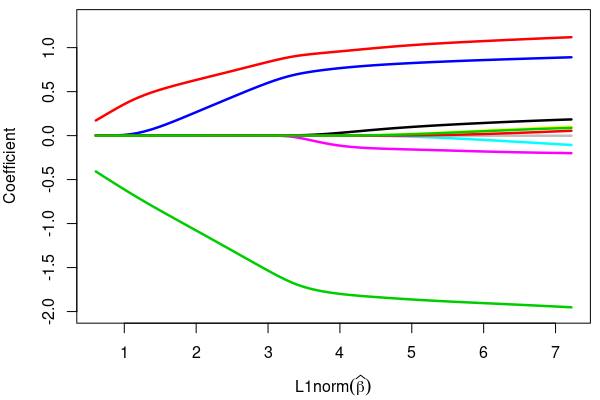}
    \end{subfigure}
    \begin{subfigure}[b]{0.4\textwidth}
        \includegraphics[width=\textwidth]{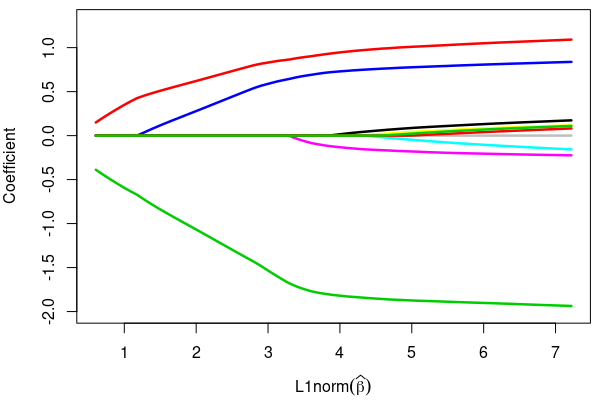}
    \end{subfigure}
    \caption{\small Solution paths of the GMS ridge (top left), the standard ridge regression (top right) and the GMS LASSO (bottom left), and  the LARS (bottom right).}\label{fig:path}
\end{figure}

Figure \ref{fig:path} shows solution-path plots that depict the relations between the tuning parameter choices and the corresponding estimated ridge and LASSO estimators. The $x$-axis indicates the $l_2$ norm of the ridge regression or $l_1$ norm of the LASSO estimators based on a series of $\lambda$'s, and the $y$-axis, the  value  of the  estimated coefficient. 
After the generator is trained by minimizing \eqref{eq:LASSO},  ridge (top left) and LASSO (bottom left) coefficient values are simply $\widehat G(\mathbbm{1},\lambda)$, which generates the curves in Figure \ref{fig:path}   by letting $\lambda$ vary from $0.0006$ to $0.6$. 
The resulting solution-paths of the GMS ridge and LASSO procedures show that the proposed method  approximates the standard ones obtained by LARS \citep{efron2004least} very accurately.

We further investigate how  the GMS-bootstrap helps to quantify uncertainty in choosing $\lambda$. 
 Figure \ref{fig:cv} illustrates some benefits of the bootstrapped CV procedure for the LASSO example.
 The left panel shows a 95\% confidence band for the CV errors across $\lambda$. As \cite{efron1997improvements} noted, the bootstrapped CV improves the performance of prediction error estimation. However, due to heavy computational burden in the standard bootstrap algorithm,  applications of the bootstrapped CV have been greatly hindered.
 The example in Figure \ref{fig:cv} shows that the GMS helps overcome this computational difficulty.
The center panel depicts the bootstrap distribution of the minimizer $\lambda$ of the CV errors (the red line is the estimated density function). 
 If 
 the CV error curve is of main interest, one can easily generate it by the GMS   using binary weights (corresponding to the chosen and left-out folds) as the input.
  In the right panel of  Figure \ref{fig:cv},  the CV error curve obtained by the standard CV computation  is nearly identical to that by the GMS.

\begin{figure}[t]
    \centering
    \begin{subfigure}[b]{0.32\textwidth}
        \includegraphics[width=\textwidth]{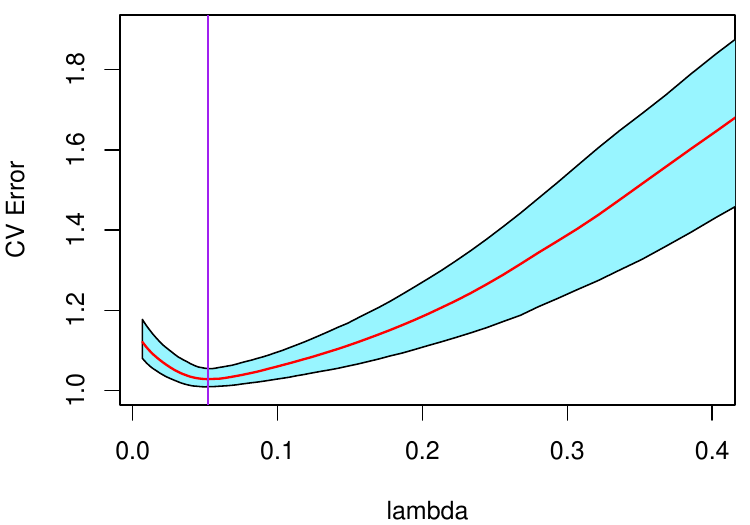}
        \label{fig:cv_err}
    \end{subfigure}
    \begin{subfigure}[b]{0.32\textwidth}
        \includegraphics[width=\textwidth]{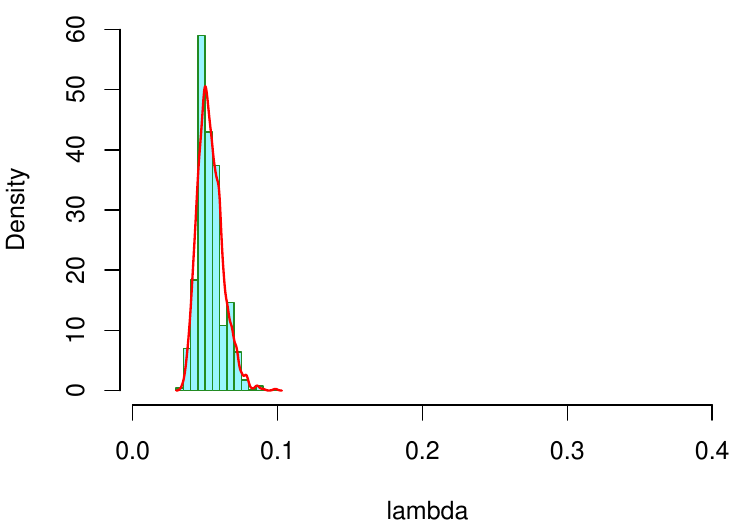}
        \label{fig:density_lam}
    \end{subfigure}
    \begin{subfigure}[b]{0.32\textwidth}
        \includegraphics[width=\textwidth]{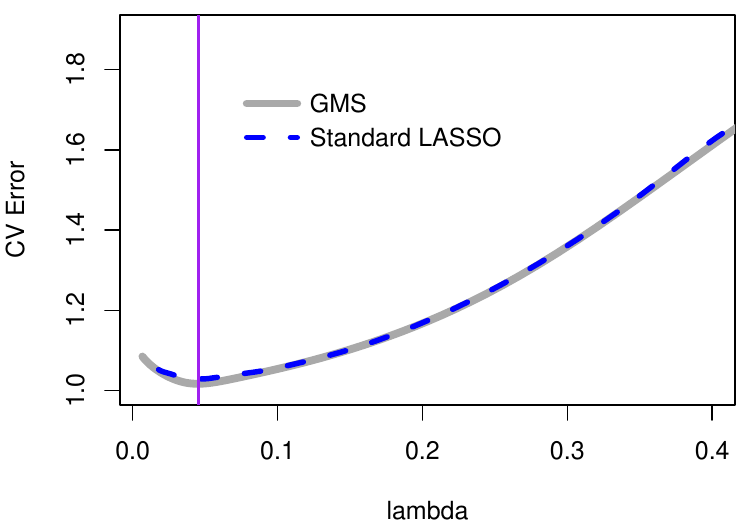}
        \label{fig:cv_err_only}
    \end{subfigure}
    \vspace{-0.3cm}
    \caption{\small Left: The 95\% confidence band of CV error evaluated from the GMS bootstrap with random weights, and the red solid line indicates the mean curve. Middle: The GMS bootstrapped distribution of the CV-error minimizer $\lambda_{\text{min}}$. Right:  CV errors based on the standard LASSO and the GMS with the constant weight vector $\mathbbm{1}$. The purple vertical line indicates the value of $\lambda$ that minimizes the CV error.
    }\label{fig:cv}
\end{figure}

\section{Quantile Regression Inference at Various Quantile Levels} \label{sec:quantilereg}
 Quantile regression models, which assume that a certain quantile of the response variable linearly depends on the covariates, have been commonly used for robust regression analysis
 \citep{yu2003quantile, yu2001bayesian,koenker2004quantile}.
 More precisely, for a given $\eta\in(0,1)$, the conditional $\eta$-th quantile of the response given $X_i$ is modeled by $X_i^\T\theta$.
The standard loss function for fitting such a model is 
%
 \begin{equation}\label{eq:QRLoss}
 \ell(\theta; y_i, X_i) = \rho_\eta(y_i - X_i^\T\theta),
 \end{equation}
 where $\rho_\eta(u) = (\eta- I(u<0))u$.   
The inference for the regression coefficients in this setting is more challenging than that for parametric regression models, because the sampling distribution of the coefficient estimates often relies on the regression error density function, which needs to be estimated and is a challenging task by itself in high-dimensional settings \citep{koenker1994confidence}. 
In routine applications of quantile regression analyses, bootstrap procedures are popular to use for  approximating the  sampling distribution of the estimates \citep{feng2011wild,hahn1995bootstrapping,kocherginsky2005practical}, which can be computationally demanding.
Furthermore, when a practitioner is interested in investigating multiple quantile levels, it is also necessary to  repeat the bootstrap procedure multiple times, each  at a different quantile level. Such a computational burden is prohibitive when the data size is large.     

\begin{figure}[t]
\vspace{-0.5cm}
\centering
        \includegraphics[width=0.6\textwidth]{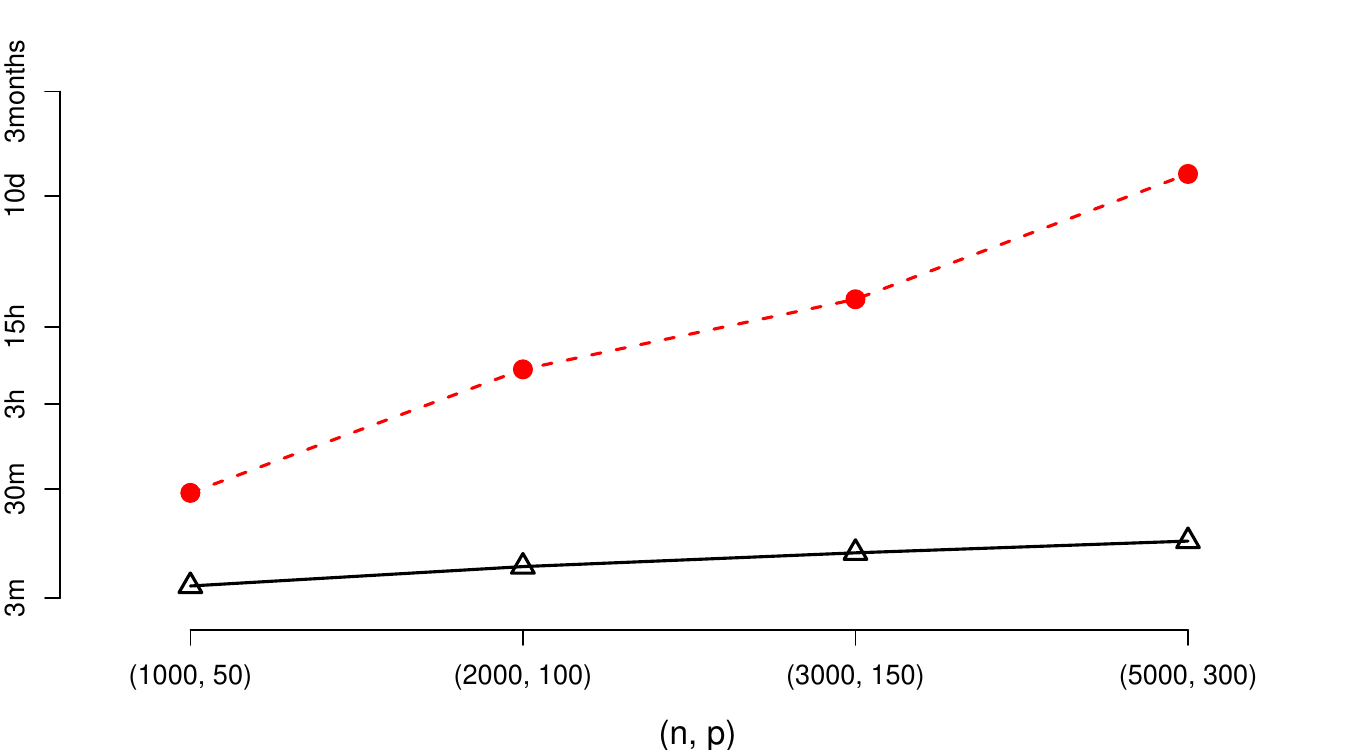}
        \caption{Computation time for the GMS quantile regression (black solid line with triangle marks) and its conventional counterpart (red dashed line with filled-dot marks).}
\end{figure}

By using $\ell(G(\bw, \eta); y_i, X_i) =\rho_\eta(y_i-X_i^\top G(\bw,\eta))$ in \eqref{eq:GMS}, we apply the GMS to  overcome the computational challenges for the  inference of quantile regression models with a GMS loss of 

\begin{equation}\label{eq:QR_GMS}
     \widehat G = \argmin_{G}\mathbb{E_{\mathbf{w},\eta}}\Big[\sum_i^n w_i \rho_\eta(y_i-X_i^\top G(\bw,\eta))\Big],
 \end{equation}
 where $\mathbb{E}_{\mathbf{w},\eta}$ is the expectation operator on $\bw$ and $\eta$, assuming that $\eta$ follows some distribution $\mathbb{P}_\eta$ whose support is (0,1) and independent with $\bw$. A default choice is to add random noises to the candidate set of quantile levels, and let $\bw$ follow the probability law in \eqref{eq:sub_assign}.


 
 To demonstrate the effectiveness of this procedure, we test the method on a simulation setting examined in \cite{feng2011wild}. The data set is generated from the model $y_i = X_i^{\top}\theta_0 + 3^{-1/2}[2+\{1+(x_{1i}-8)^2 +x_{2i}\}/10]\epsilon_i$, $i=1,\ldots,n$,  where  $X_i=(x_{i1},\ldots,x_{ip})^\top$, $n=500$, $p=5$, $\theta_0=\mathbbm{1}_5^\top$, and $\epsilon_i\sim t_3$. We let $x_{2i}=1$ for $i\leq 400$ and $=0$  for $i>400$, and generate  the other covariates independently from the standard log-normal distribution. 
Figure \ref{fig:QR} (b)--(d) compare the $90\%$ confidence  bands of several coefficients generated by the GMS with those obtained by the standard bootstrap over quantiles varying from $0.05$ and $0.95$, showing that the the approaches result in nearly identical bands. 


\FloatBarrier
 \begin{figure}[htp]
    \centering

    \begin{subfigure}[b]{0.7\textwidth}
        \includegraphics[width=\textwidth]{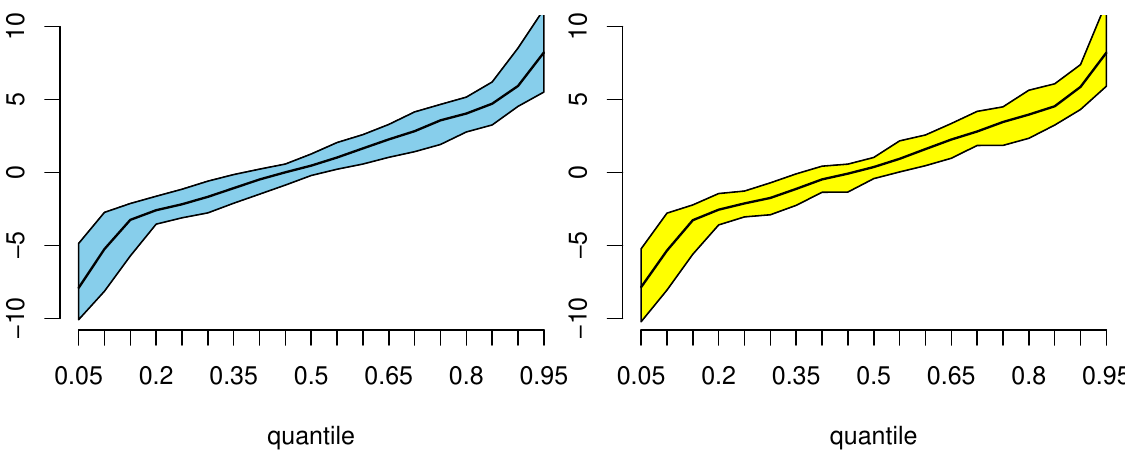}
        \subcaption{Intercept}
    \end{subfigure}
    \begin{subfigure}[b]{0.7\textwidth}
        \includegraphics[width=\textwidth]{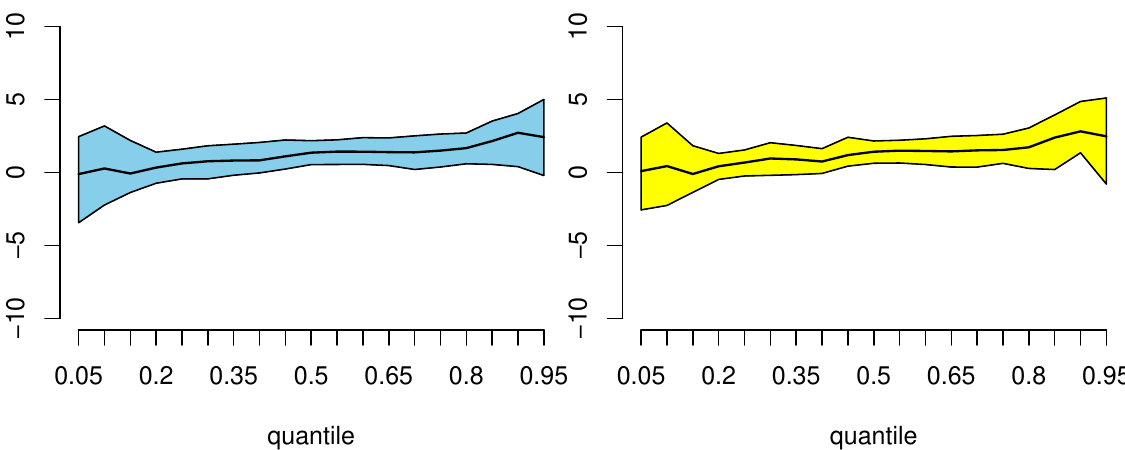}
        \subcaption{$\theta_3$}
    \end{subfigure}
    \begin{subfigure}[b]{0.7\textwidth}
        \includegraphics[width=\textwidth]{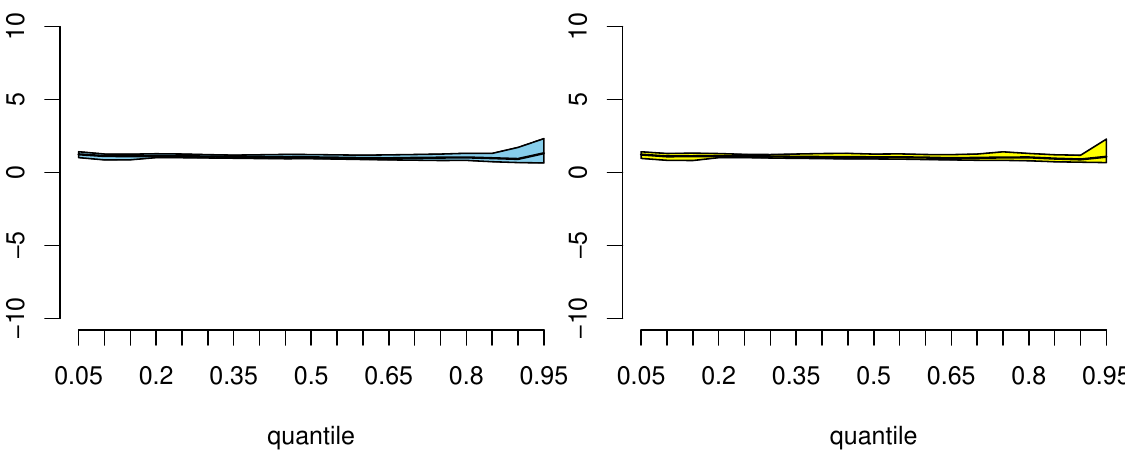}
        \subcaption{$\theta_5$}
    \end{subfigure}
    
    \caption{\small (a):  The red dots and  black triangles depict the computation times required by the standard bootstrap and the GMS, respectively, in generating 5,000 bootstrap estimators. (b)--(d): Comparisons between the 90\% confidence bands obtained from the GMS (blue) and the classical bootstrap (yellow) across quantile levels ranging from 5\% to 95\%. The \texttt{quantreg} R package is used for the conventional. }\label{fig:QR}
\end{figure}

To investigate  computational efficiency of the GMS for quantile regression, we increase the sample size and the number of predictors in the above simulation model to  $(n,p)=(1000,50)$, $(2000,100)$, $(3000,150)$, and $(5000,300)$, respectively, and consider quantile levels varying from $0.05$ to $0.95$ with a skip of $0.05$ (total $19$ quantile levels). We set the first five coefficients of $\theta_0$ to be one and the others be zero. 
Our target is to obtain $5,000$ bootstrap samples under each setting. Due to heavy computational burden of the standard bootstrap procedure, we compute only five bootstrap evaluations and report an estimated time from them (e.g., multiplying $1,000$ to the time taken for the five evaluations).   Figure \ref{fig:QR} (a) depicts the computation time required for each procedure. While the GMS can be trained in less than $10$ minutes for moderately large data size ($n=5000,p=300$), the standard bootstrap requires more than $30$ minutes for  the smallest data set ($n=1000,p=50$) and about $3$ months for the case of $(n=5000, p=300)$.
\section{Computational Strategies for Training the Generator}\label{sec:computation}
\subsection{Multilayer perceptron}\label{sec:NN}
 \emph{Neural networks} have been shown  effective for approximating  functions with  complicated structures.  
Recently, researchers have experimented with various novel ways of using neural networks, such as constructing   generators of real-life-like images and creating generative adversarial networks for approximating high-dimensional distributions 
\citep{ledig2017photo, wang2018high, karras2018progressive, goodfellow2014generative, arjovsky2017wasserstein}.    
The simplest neural network  structure is 
a class of MLPs/FNNs constructed by composing activated linear transformations. For $k=1,\dots,K$, let $g_k$ denote the feed-forward mapping represented by $N^{(k)}$ hidden nodes,  where $g_k: \mathbb{R}^{N^{(k)}}\mapsto \mathbb{R}^{N^{(k+1)}}$ is defined as $g_k(\bX)= \sigma(\bU^{(k)}\bX +\bb^{(k)})\in\mathbb{R}^{N^{(k+1)}}
$,
 where $\bX\in\mathbb{R}^{N^{(k)}}$ is the input variable of $g_k$. Also, this function is characterized by a ``weight'' parameter and a ``bias'' parameter: 
the $N^{(k+1)}\times N^{(k)}$ weight matrix $\bU^{(k)}$ and  the $N^{(k+1)}$-dimensional bias vector  $\bb^{(k)}=\{b_1^{(k)},\dots b_{N^{(k)}}^{(k)}\}$. 
A $K$-layer  MLP function $g:\mathbb{R}^{N^{(1)}}\mapsto \mathbb{R}^{D}$ can be defined by the composition of these functions as 
\begin{equation}\label{eq:MLP}
    g(\bX) =  L \circ g_{K} \circ \dots \circ g_{1}(\bX),
\end{equation} 
where $L: \mathbb{R}^{N^{(K)}}\mapsto \mathbb{R}^{D}$ is a linear function that maps the final hidden layer $g_K\circ \dots \circ g_1(\bX)$ to the $D$-dimensional output space of $g$. Commonly used activation functions include the sigmoid function, the hyperbolic tangent function,  the {Rectified Linear Unit} (ReLU) \citep{nair2010rectified}, the {Exponential Linear Unit} \citep{clevert2015fast},  the {Gaussian Error Linear Unit}  \citep{hendrycks2016gaussian}, etc.
We here employ neural networks with the ReLU activation function $\sigma(t) = \max\{t,0\}$
to construct  generator $G$ in \eqref{eq:GMS} in a novel way as characterized by the integrative loss  (\ref{eq:GMS}) and the weight multiplicative MLP explained below.

\FloatBarrier
\begin{figure}[t]
    \centering
        \includegraphics[width=0.7\textwidth]{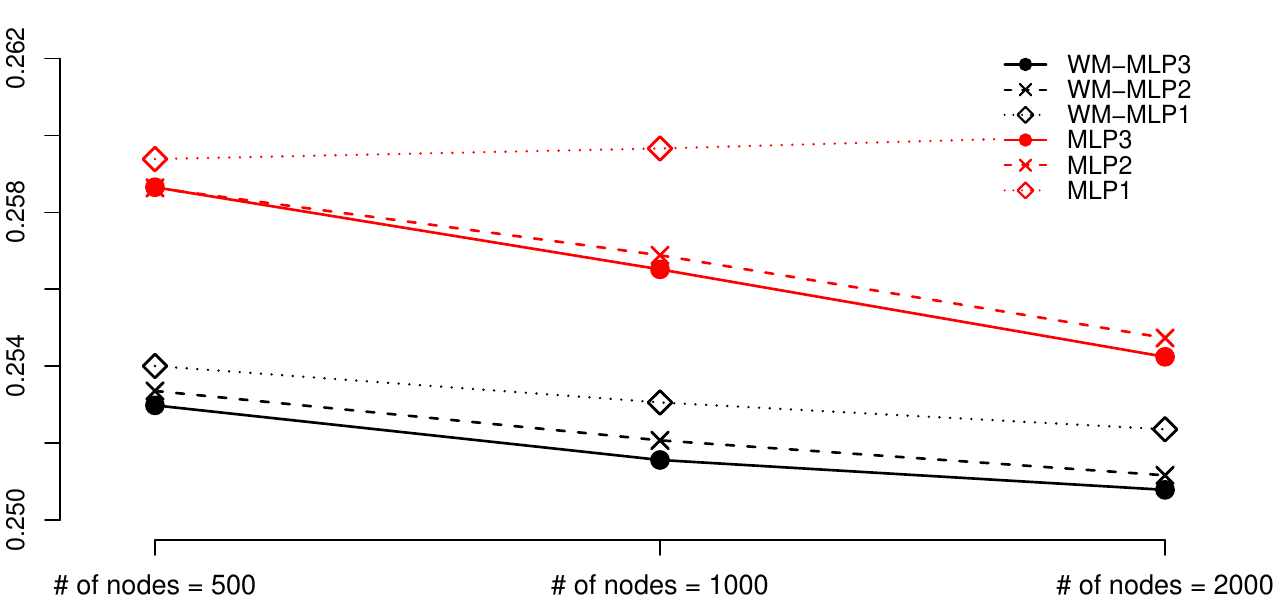}
 \vspace{0.5cm}
\caption{\small Comparison of the losses obtained by the simple MLP and the WM-MLP with various numbers of hidden layers and nodes. The number noted after "MLP" indicates the number of layers $K$. }\label{fig:NNcomparison}
\end{figure}

\subsection{Weight multiplicative MLP}

Despite its generalizability and practicability, we observe that the simple MLP converges slowly for our GMS  applications (as shown in Figure~\ref{fig:NNcomparison}). 
We propose a modification motivated by the Taylor approximation of the first derivative of the weighted loss function. For illustration, let us consider the weighted M-estimation loss $\sum_{i=1}^nw_i\ell(\theta;y_i)$ and its optimizer $\hat\theta_\bw$ in \eqref{eq:loss} for a case of $p=1$ (ignoring $\eta$ and $\lambda$ for simplicity). Under mild conditions, we assume that $\sum_{i=1}^nw_i\ell'(\hat\theta_{\bw};y_i)=0$, where $\ell'$ is the first  derivative of $\ell$ with respect to $\theta$. Then, by using a Taylor approximation of $\ell'$ at a local region of some arbitrary $g(\bw)$, we obtain that 
\[0  =\sum_{i=1}^nw_i\ell'(\hat\theta_{\bw};y_i)\approx \sum_{i=1}^nw_i\ell'(g(\bw);y_i) + \sum_{i=1}^nw_i\ell''(g(\bw);y_i)(\hat\theta_{\bw} - g(\bw)),\] 
where $\ell''(\theta,y)$ denotes the second derivative of $\ell$ with respect to $\theta$.
Thus, we have 
\begin{equation}\label{eq:taylor} 
\hat\theta_\bw \approx g(\bw)-\sum_{i=1}^n \frac{w_i \ell'(g(\bw);y_i)}{\sum_{j=1}^n w_j \ell''(g(\bw);y_j) }    \stackrel{\Delta}{=} g(\bw) + \sum_{i=1}^nw_i h_i(\bw).
\end{equation}
Motivated by this approximation, we propose a new neural network  structure called the \emph{Weight Multiplicative MLP} (WM-MLP) as the sum of a simple MLP and a weight multiplicative one: 
\begin{equation}\label{eq:WMLP}
G(\bw,\lambda,\eta) = \underbrace{L_1\circ B_K(\bw, \lambda, \eta)}_{\substack{\text{Simple MLP: } \\
g(\bw)}
} + \underbrace{L_2\circ(\{f \circ B_K(\bw, \lambda, \eta)\}\odot \bw)}_{\substack{\text{Weight multiplicative network: } \\  \sum_{i=1}^nw_ih_i(\bw,\lambda,\eta)  }},
 \end{equation}
where ``$\odot$'' indicates an element-wise multiplication operator; $L_1:\mathbb{R}^{H}\mapsto \mathbb{R}^{p}$ and $L_2:\mathbb{R}^{n}\mapsto \mathbb{R}^{p}$ are linear functions;
  $B_K: \mathbb{R}^{n+1+1} \mapsto \mathbb{R}^{H}$  and $f:\mathbb{R}^{H} \mapsto \mathbb{R}^{n}$ are simple MLPs with $K$ hidden layers and one hidden layer, respectively. For a  large $n$, 
the subgroup bootstrap in Section \ref{sec:Blockbootstrap} reduces  the dimension of $\bw$ and the network size. 

To demonstrate the improvement, we compare the performances of WM-MLP and the simple MLP for various sizes of hidden nodes $(500, 1000, 2000)$ and layers ($K=1,2,3$), for a logistic regression example.
 The true $\theta$'s in the simulations are equi-spaced  between $-0.5$ and $0.5$ with  $p=100$ and $n=1000$. We train the generator $G$ from ten random initializations and report  the average  loss values after 30,000 iterative updates for each MLP structure. The results are summarized in Figure \ref{fig:NNcomparison},  demonstrating that for all network sizes the proposed WM-MLP outperforms the simple MLP uniformly. In comparison to   a large-sized MLP with three hidden layers and $2000$ neurons, even a small-scale WM-MLP with a single hidden layer and 500 neurons achieves a lower loss, whereas the simple MLP with one hidden layer performs much poorly. 
For all examples in the paper, we used the WM-MLP with three hidden layers as a default, and observed that the resulting generator function based on the WM-MLP performed satisfactorily.     

\begin{algorithm}[t]
\footnotesize
\caption{\footnotesize    A general algorithm to train the GMS. }\label{alg:alg}
\begin{algorithmic}
\STATE $\bullet$ Set $\mathbb{P}_{\balpha,\lambda,\eta}$, $S$ (subgroup size), $M$ (Monte Carlo sample size), and $T$ (total iterations).
\STATE $\bullet$ Randomly split  the full data into $S$ subgroups, resulting in an index function $h(\cdot)$ in \eqref{eq:sub_assign}.
\STATE $\bullet$ Initialize the neural net parameter $\phi^{(0)}$.
\STATE $\bullet$ Set $t=0$.
\WHILE{the stop condition is not satisfied or $t<T$} 
\STATE $\bullet$ Independently sample  $M$ values of $\balpha$'s,  $\lambda$'s, and $\eta$'s from $\mathbb{P}_{\balpha,\lambda,\eta}$.
    \STATE $\bullet$ Consider  $L =  \frac{1}{ M}\sum_{m=1}^M\sum_{i=1}^n\alpha_{h(i)}^{(m)}l(G_{\phi^{(t)}}(\balpha^{(m)},\lambda^{(m)},\eta^{(m)}); y_i)/n+\lambda^{(m)}u(G_{\phi^{(t)}}(\balpha^{(m)},\lambda^{(m)},\eta^{(m)}))$, 
    \STATE where $\balpha^{(m)}$ is the $m$-th sample of $M$ $\balpha$'s.
    \STATE $\bullet$ Update $\phi^{(t+1)}$ by using the gradient of $L$ via a SGD step.
    \STATE $\bullet$ Let $t=t+1$.
\ENDWHILE

\end{algorithmic}
\end{algorithm}

\subsection{Computational strategy in optimization}\label{sec:comp}
It is straightforward to optimize the GMS integrative loss \eqref{eq:GMS} because the expectation can be approximated by a few Monte Carlo samples at each iteration. We  use a variant of the popular SGD algorithms such as \emph{Adam} \citep{kingma2014adam}, \emph{AdaGrad} \citep{duchi2011adaptive}, \emph{RMSProp} \citep{tieleman2012lecture}, etc, to iteratively update the neural net parameters until the algorithm  converges. Algorithm \ref{alg:alg} summarizes the detailed steps of the GMS. As in \eqref{eq:MC_Approx}, this algorithm samples $M$ values of $\bw$'s and $\lambda$'s to approximate the expectation and updates the neural network  parameters via SGD. 
It is not uncommon nowadays for a data set to be extremely large, to the point that the full data size surpasses the memory capacity of the computer in use. Data subsampling would be advantageous in this setting for training the GMS, which partially updates the weights corresponding to the subsampled data in the same spirit as stochastic optimization \citep{allen2019convergence}. \\



\noindent{\bf Technical details of the optimization.} In all our examples, 
we use the WM-MLP with three hidden layers and 1,000 hidden neurons in each layer. In \texttt{Pytorch}, algorithm Adam  is used with a learning rate of $0.0003$  and a decay rate of $t^{-0.3}$ by default. We use full samples in the SGD optimization without mini-batches because the data sizes of the examples we considered  are manageable. However, when the data size is massive, minibatch subsampling would be necessary.   \\

\noindent{\bf Choosing  distributions for $\bw$, $\lambda$, and $\eta$.} 
For bootstrap procedures, the distribution of bootstrap weights $\bw$ (or $\balpha$) can be easily chosen
depending on the practitioner's interest; e.g.,  $\bw \sim \text{Multinomial}(n, \mathbbm{1}_n/n)\ \text{ or } \  \bw \sim n\times\text{Dirichlet}(n,\mathbbm{1}_n)$. When $n$ is excessively large, the dimension of $\bw$ can be reduced by   the subgroup bootstrapping method in Section \ref{sec:Blockbootstrap}. As a general rule, when $n>500$, we recommend considering subgrouping. While our theoretical evidence  suggests that $S \succ n^{1/2}$ is optimal (see Section \ref{sec:theory_subgroup} in Supplementary Materials), empirically setting $S$ to a few hundreds performs well in all situations shown in this paper. By default, $S=100$ was used.  
 Choosing the training distributions for $\lambda$ and $\eta$ is more arbitrary because usually we have no reference distributions for $\lambda$ and $\eta$  unlike the case of $\bw$. We may first set candidate sets for $\lambda$ and $\eta$ in advance (which can be large in size) and then add some random noises to form mixture distributions. For example, we can generate  $\lambda = \exp\{\log \lambda' + \epsilon \} $, where $\lambda'$ is randomly selected from the candidate set and $\epsilon \sim N(0,\delta)$ with $\delta = 0.2^2$ as default. 
 For the quantile regression example in Section \ref{sec:quantilereg}, we generate $\eta=\eta' + N(0,0.03^2)$ with $\eta'$  randomly selected from a  pre-determined candidate set, and then truncated to be 
 in $(0.001,0.999)$.    \\

\noindent{\bf Training stopping criteria.} { 
In order to judge the  convergence  in training the generator function, we first set the maximum number of epochs depending on computational resources at hands (our default is 20,000 epochs). In addition to this stopping criterion, we also  consider an {\it early stopping rule} that has been commonly used in  training general neural networks \citep{heckel2021early,li2020gradient,prechelt1998early} to determine when we stop the optimization algorithm before reaching the maximum number of epochs. Intuitively, we stop the algorithm when the updates do not further reduce the loss value. More specifically, for each epoch $t$, we evaluate the averaged loss value $L^t$ on epoch $t$ and compare it with those of the previous epochs $\{L^{t-\ell}, \ell=1,2,\ldots, k\}$ for some lags. We terminate the SGD algorithm  if $L^{t}$ is within $\epsilon$ of a quantile (such as the median) of the previous losses.
We recommend to monitor the change of loss values in the previous $k$=100 epochs,
and use the 25th percentile with $\epsilon=0.01$.
}

\section{Conclusion}\label{sec:conclusion}
We propose the GMS as a general computational framework to accelerate repeated calculations for (penalized) weighted M-estimations. The GMS was shown effective for  a variety of statistical inference procedures, including bootstrap methods and  cross-validations for general M-estimators. We apply the GMS to   a variety of models, including LASSO,  logistic regression, quantile regression, etc. The GMS performs well in all of the situations we investigated, and the weighted M-estimators generated by the GMS are sufficiently accurate  and comparable to  the much more computationally expensive traditional solutions for all inference purposes. By lowering the computational barrier associated with  repetitious data-splitting or data-sampling processes such as (bootstrapped) CVs and iterated bootstrap, the GMS opens up a new perspective on modern statistics. To date, these approaches have been less noticed and rarely practiced by the statistical community not because they are less valuable, but because their computation cost is prohibitively high. We expect that the GMS will prove to be an effective tool for augmenting the power of statistical models in the era of big data.


\newpage
\begingroup
\setstretch{1.5}
    \bibliography{myReference.bib}

\begin{thebibliography}{}

\bibitem[\protect\citeauthoryear{Allen-Zhu, Li, and Song}{Allen-Zhu et~al.}{2019}]{allen2019convergence}
Allen-Zhu, Z., Y.~Li, and Z.~Song (2019).
\newblock A convergence theory for deep learning via over-parameterization.
\newblock In {\em International Conference on Machine Learning}, pp.\  242--252. PMLR.

\bibitem[\protect\citeauthoryear{Arjovsky, Chintala, and Bottou}{Arjovsky et~al.}{2017}]{arjovsky2017wasserstein}
Arjovsky, M., S.~Chintala, and L.~Bottou (2017).
\newblock Wasserstein generative adversarial networks.
\newblock In {\em International Conference on Machine Learning}, pp.\  214--223.

\bibitem[\protect\citeauthoryear{Barbe and Bertail}{Barbe and Bertail}{2012}]{barbe2012weighted}
Barbe, P. and P.~Bertail (2012).
\newblock {\em The weighted bootstrap}, Volume~98.
\newblock Springer Science \& Business Media.

\bibitem[\protect\citeauthoryear{Chatterjee, Bose, et~al.}{Chatterjee et~al.}{2005}]{chatterjee2005generalized}
Chatterjee, S., A.~Bose, et~al. (2005).
\newblock Generalized bootstrap for estimating equations.
\newblock {\em The Annals of Statistics\/}~{\em 33\/}(1), 414--436.

\bibitem[\protect\citeauthoryear{Cheng and Huang}{Cheng and Huang}{2010}]{cheng2010bootstrap}
Cheng, G. and J.~Z. Huang (2010).
\newblock Bootstrap consistency for general semiparametric m-estimation.
\newblock {\em The Annals of Statistics\/}~{\em 38\/}(5), 2884--2915.

\bibitem[\protect\citeauthoryear{Clevert, Unterthiner, and Hochreiter}{Clevert et~al.}{2015}]{clevert2015fast}
Clevert, D.-A., T.~Unterthiner, and S.~Hochreiter (2015).
\newblock Fast and accurate deep network learning by exponential linear units (elus).
\newblock {\em arXiv preprint arXiv:1511.07289\/}.

\bibitem[\protect\citeauthoryear{Cybenko}{Cybenko}{1989}]{cybenko1989approximation}
Cybenko, G. (1989).
\newblock Approximation by superpositions of a sigmoidal function.
\newblock {\em Mathematics of control, signals and systems\/}~{\em 2\/}(4), 303--314.

\bibitem[\protect\citeauthoryear{Diciccio and Efron}{Diciccio and Efron}{1992}]{diciccio1992more}
Diciccio, T. and B.~Efron (1992).
\newblock More accurate confidence intervals in exponential families.
\newblock {\em Biometrika\/}~{\em 79\/}(2), 231--245.

\bibitem[\protect\citeauthoryear{Duchi, Hazan, and Singer}{Duchi et~al.}{2011}]{duchi2011adaptive}
Duchi, J., E.~Hazan, and Y.~Singer (2011).
\newblock Adaptive subgradient methods for online learning and stochastic optimization.
\newblock {\em Journal of machine learning research\/}~{\em 12\/}(7).

\bibitem[\protect\citeauthoryear{Efron}{Efron}{1979}]{efron1979bootstrap}
Efron, B. (1979).
\newblock Bootstrap methods: Another look at the jackknife.
\newblock {\em The Annals of Statistics\/}~{\em 7\/}(1), 1--26.

\bibitem[\protect\citeauthoryear{Efron}{Efron}{1987}]{efron1987better}
Efron, B. (1987).
\newblock Better bootstrap confidence intervals.
\newblock {\em Journal of the American statistical Association\/}~{\em 82\/}(397), 171--185.

\bibitem[\protect\citeauthoryear{Efron, Hastie, Johnstone, and Tibshirani}{Efron et~al.}{2004}]{efron2004least}
Efron, B., T.~Hastie, I.~Johnstone, and R.~Tibshirani (2004).
\newblock Least angle regression.
\newblock {\em The Annals of statistics\/}~{\em 32\/}(2), 407--499.

\bibitem[\protect\citeauthoryear{Efron and Tibshirani}{Efron and Tibshirani}{1997}]{efron1997improvements}
Efron, B. and R.~Tibshirani (1997).
\newblock Improvements on cross-validation: the 632+ bootstrap method.
\newblock {\em Journal of the American Statistical Association\/}~{\em 92\/}(438), 548--560.

\bibitem[\protect\citeauthoryear{Efron and Tibshirani}{Efron and Tibshirani}{1994}]{efron1994introduction}
Efron, B. and R.~J. Tibshirani (1994).
\newblock {\em An introduction to the bootstrap}.
\newblock CRC press.

\bibitem[\protect\citeauthoryear{Feng, He, and Hu}{Feng et~al.}{2011}]{feng2011wild}
Feng, X., X.~He, and J.~Hu (2011).
\newblock Wild bootstrap for quantile regression.
\newblock {\em Biometrika\/}~{\em 98\/}(4), 995--999.

\bibitem[\protect\citeauthoryear{Geer, van~de Geer, and Williams}{Geer et~al.}{2000}]{geer2000empirical}
Geer, S.~A., S.~van~de Geer, and D.~Williams (2000).
\newblock {\em Empirical Processes in M-estimation}, Volume~6.
\newblock Cambridge university press.

\bibitem[\protect\citeauthoryear{Goodfellow, Pouget-Abadie, Mirza, Xu, Warde-Farley, Ozair, Courville, and Bengio}{Goodfellow et~al.}{2014}]{goodfellow2014generative}
Goodfellow, I., J.~Pouget-Abadie, M.~Mirza, B.~Xu, D.~Warde-Farley, S.~Ozair, A.~Courville, and Y.~Bengio (2014).
\newblock Generative adversarial nets.
\newblock In {\em Advances in neural information processing systems}, pp.\  2672--2680.

\bibitem[\protect\citeauthoryear{Hahn}{Hahn}{1995}]{hahn1995bootstrapping}
Hahn, J. (1995).
\newblock Bootstrapping quantile regression estimators.
\newblock {\em Econometric Theory\/}~{\em 11\/}(1), 105--121.

\bibitem[\protect\citeauthoryear{Hall}{Hall}{1986}]{hall1986bootstrap}
Hall, P. (1986).
\newblock On the bootstrap and confidence intervals.
\newblock {\em The Annals of Statistics\/}, 1431--1452.

\bibitem[\protect\citeauthoryear{Hall}{Hall}{1988}]{hall1988theoretical}
Hall, P. (1988).
\newblock Theoretical comparison of bootstrap confidence intervals.
\newblock {\em The Annals of Statistics\/}, 927--953.

\bibitem[\protect\citeauthoryear{Hall}{Hall}{1992}]{hall1992bootstrap}
Hall, P. (1992).
\newblock On bootstrap confidence intervals in nonparametric regression.
\newblock {\em The Annals of Statistics\/}, 695--711.

\bibitem[\protect\citeauthoryear{Hall}{Hall}{2013}]{hall2013bootstrap}
Hall, P. (2013).
\newblock {\em The bootstrap and Edgeworth expansion}.
\newblock Springer Science \& Business Media.

\bibitem[\protect\citeauthoryear{Hall and Martin}{Hall and Martin}{1988}]{hall1988bootstrap}
Hall, P. and M.~A. Martin (1988).
\newblock On bootstrap resampling and iteration.
\newblock {\em Biometrika\/}~{\em 75\/}(4), 661--671.

\bibitem[\protect\citeauthoryear{H{\"a}rdle, Horowitz, and Kreiss}{H{\"a}rdle et~al.}{2003}]{hardle2003bootstrap}
H{\"a}rdle, W., J.~Horowitz, and J.-P. Kreiss (2003).
\newblock Bootstrap methods for time series.
\newblock {\em International Statistical Review\/}~{\em 71\/}(2), 435--459.

\bibitem[\protect\citeauthoryear{Heckel and Yilmaz}{Heckel and Yilmaz}{2021}]{heckel2021early}
Heckel, R. and F.~F. Yilmaz (2021).
\newblock Early stopping in deep networks: Double descent and how to eliminate it.
\newblock In {\em International Conference on Learning Representations}.

\bibitem[\protect\citeauthoryear{Hendrycks and Gimpel}{Hendrycks and Gimpel}{2016}]{hendrycks2016gaussian}
Hendrycks, D. and K.~Gimpel (2016).
\newblock Gaussian error linear units (gelus).
\newblock {\em arXiv preprint arXiv:1606.08415\/}.

\bibitem[\protect\citeauthoryear{Hoerl and Kennard}{Hoerl and Kennard}{1970}]{hoerl1970ridge}
Hoerl, A.~E. and R.~W. Kennard (1970).
\newblock Ridge regression: Biased estimation for nonorthogonal problems.
\newblock {\em Technometrics\/}~{\em 12\/}(1), 55--67.

\bibitem[\protect\citeauthoryear{Huber}{Huber}{1992}]{huber1992robust}
Huber, P.~J. (1992).
\newblock Robust estimation of a location parameter.
\newblock In {\em Breakthroughs in statistics}, pp.\  492--518. Springer.

\bibitem[\protect\citeauthoryear{Karras, Aila, Laine, and Lehtinen}{Karras et~al.}{2018}]{karras2018progressive}
Karras, T., T.~Aila, S.~Laine, and J.~Lehtinen (2018).
\newblock Progressive growing of gans for improved quality, stability, and variation.
\newblock In {\em International Conference on Learning Representations}.

\bibitem[\protect\citeauthoryear{Kingma and Ba}{Kingma and Ba}{2014}]{kingma2014adam}
Kingma, D.~P. and J.~Ba (2014).
\newblock Adam: A method for stochastic optimization.
\newblock {\em arXiv preprint arXiv:1412.6980\/}.

\bibitem[\protect\citeauthoryear{Kleiner, Talwalkar, Sarkar, and Jordan}{Kleiner et~al.}{2014}]{kleiner2014scalable}
Kleiner, A., A.~Talwalkar, P.~Sarkar, and M.~I. Jordan (2014).
\newblock A scalable bootstrap for massive data.
\newblock {\em Journal of the Royal Statistical Society: Series B (Statistical Methodology)\/}~{\em 76\/}(4), 795--816.

\bibitem[\protect\citeauthoryear{Kocherginsky, He, and Mu}{Kocherginsky et~al.}{2005}]{kocherginsky2005practical}
Kocherginsky, M., X.~He, and Y.~Mu (2005).
\newblock Practical confidence intervals for regression quantiles.
\newblock {\em Journal of Computational and Graphical Statistics\/}~{\em 14\/}(1), 41--55.

\bibitem[\protect\citeauthoryear{Koenker}{Koenker}{1994}]{koenker1994confidence}
Koenker, R. (1994).
\newblock Confidence intervals for regression quantiles.
\newblock In {\em Asymptotic statistics}, pp.\  349--359. Springer.

\bibitem[\protect\citeauthoryear{Koenker}{Koenker}{2004}]{koenker2004quantile}
Koenker, R. (2004).
\newblock Quantile regression for longitudinal data.
\newblock {\em Journal of Multivariate Analysis\/}~{\em 91\/}(1), 74--89.

\bibitem[\protect\citeauthoryear{Kosorok}{Kosorok}{2008}]{kosorok2008m}
Kosorok, M.~R. (2008).
\newblock M-estimators.
\newblock {\em Introduction to Empirical Processes and Semiparametric Inference\/}, 263--282.

\bibitem[\protect\citeauthoryear{Lahiri}{Lahiri}{1999}]{lahiri1999theoretical}
Lahiri, S.~N. (1999).
\newblock Theoretical comparisons of block bootstrap methods.
\newblock {\em Annals of Statistics\/}, 386--404.

\bibitem[\protect\citeauthoryear{Ledig, Theis, Husz{\'a}r, Caballero, Cunningham, Acosta, Aitken, Tejani, Totz, Wang, et~al.}{Ledig et~al.}{2017}]{ledig2017photo}
Ledig, C., L.~Theis, F.~Husz{\'a}r, J.~Caballero, A.~Cunningham, A.~Acosta, A.~Aitken, A.~Tejani, J.~Totz, Z.~Wang, et~al. (2017).
\newblock Photo-realistic single image super-resolution using a generative adversarial network.
\newblock In {\em Proceedings of the IEEE conference on computer vision and pattern recognition}, pp.\  4681--4690.

\bibitem[\protect\citeauthoryear{Lee and Young}{Lee and Young}{1995}]{lee1995asymptotic}
Lee, S.~M. and G.~A. Young (1995).
\newblock Asymptotic iterated bootstrap confidence intervals.
\newblock {\em The Annals of Statistics\/}, 1301--1330.

\bibitem[\protect\citeauthoryear{Lee and Young}{Lee and Young}{1999}]{lee1999effect}
Lee, S.~M. and G.~A. Young (1999).
\newblock The effect of monte carlo approximation on coverage error of double-bootstrap confidence intervals.
\newblock {\em Journal of the Royal Statistical Society: Series B (Statistical Methodology)\/}~{\em 61\/}(2), 353--366.

\bibitem[\protect\citeauthoryear{Li, Soltanolkotabi, and Oymak}{Li et~al.}{2020}]{li2020gradient}
Li, M., M.~Soltanolkotabi, and S.~Oymak (2020).
\newblock Gradient descent with early stopping is provably robust to label noise for overparameterized neural networks.
\newblock In {\em International Conference on Artificial Intelligence and Statistics}, pp.\  4313--4324. PMLR.

\bibitem[\protect\citeauthoryear{Lu, Pu, Wang, Hu, and Wang}{Lu et~al.}{2017}]{lu2017expressive}
Lu, Z., H.~Pu, F.~Wang, Z.~Hu, and L.~Wang (2017).
\newblock The expressive power of neural networks: A view from the width.
\newblock In {\em Proceedings of the 31st International Conference on Neural Information Processing Systems}, pp.\  6232--6240.

\bibitem[\protect\citeauthoryear{Martin}{Martin}{1992}]{martin1992double}
Martin, M.~A. (1992).
\newblock On the double bootstrap.
\newblock In {\em Computing science and statistics}, pp.\  73--78. Springer.

\bibitem[\protect\citeauthoryear{McCarthy, Zhang, Brown, Berk, Buja, George, and Zhao}{McCarthy et~al.}{2018}]{mccarthy2018calibrated}
McCarthy, D., K.~Zhang, L.~D. Brown, R.~Berk, A.~Buja, E.~I. George, and L.~Zhao (2018).
\newblock Calibrated percentile double bootstrap for robust linear regression inference.
\newblock {\em Statistica Sinica\/}~{\em 28\/}(4), 2565--2589.

\bibitem[\protect\citeauthoryear{Nair and Hinton}{Nair and Hinton}{2010}]{nair2010rectified}
Nair, V. and G.~E. Hinton (2010).
\newblock Rectified linear units improve restricted boltzmann machines.
\newblock In {\em Proceedings of the 27th International Conference on International Conference on Machine Learning}, pp.\  807--814.

\bibitem[\protect\citeauthoryear{Newton and Raftery}{Newton and Raftery}{1994}]{newton1994approximate}
Newton, M.~A. and A.~E. Raftery (1994).
\newblock Approximate {B}ayesian inference with the weighted likelihood bootstrap.
\newblock {\em Journal of the Royal Statistical Society: Series B (Methodological)\/}~{\em 56\/}(1), 3--26.

\bibitem[\protect\citeauthoryear{Pr{\ae}stgaard and Wellner}{Pr{\ae}stgaard and Wellner}{1993}]{praestgaard1993exchangeably}
Pr{\ae}stgaard, J. and J.~A. Wellner (1993).
\newblock Exchangeably weighted bootstraps of the general empirical process.
\newblock {\em The Annals of Probability\/}, 2053--2086.

\bibitem[\protect\citeauthoryear{Prechelt}{Prechelt}{1998}]{prechelt1998early}
Prechelt, L. (1998).
\newblock Early stopping-but when?
\newblock In {\em Neural Networks: Tricks of the trade}, pp.\  55--69. Springer.

\bibitem[\protect\citeauthoryear{Rubin}{Rubin}{1981}]{rubin1981bayesian}
Rubin, D.~B. (1981).
\newblock The {B}ayesian bootstrap.
\newblock {\em The Annals of Statistics\/}~{\em 9\/}(1), 130434.

\bibitem[\protect\citeauthoryear{Rumelhart, Hinton, and Williams}{Rumelhart et~al.}{1986}]{rumelhart1986learning}
Rumelhart, D.~E., G.~E. Hinton, and R.~J. Williams (1986).
\newblock Learning representations by back-propagating errors.
\newblock {\em nature\/}~{\em 323\/}(6088), 533--536.

\bibitem[\protect\citeauthoryear{Tibshirani}{Tibshirani}{1996}]{tibshirani1996regression}
Tibshirani, R. (1996).
\newblock Regression shrinkage and selection via the lasso.
\newblock {\em J. R. Statist. Soc. B\/}, 267--288.

\bibitem[\protect\citeauthoryear{Tieleman, Hinton, et~al.}{Tieleman et~al.}{2012}]{tieleman2012lecture}
Tieleman, T., G.~Hinton, et~al. (2012).
\newblock Lecture 6.5-rmsprop: Divide the gradient by a running average of its recent magnitude.
\newblock {\em COURSERA: Neural networks for machine learning\/}~{\em 4\/}(2), 26--31.

\bibitem[\protect\citeauthoryear{Wang, Liu, Zhu, Tao, Kautz, and Catanzaro}{Wang et~al.}{2018}]{wang2018high}
Wang, T.-C., M.-Y. Liu, J.-Y. Zhu, A.~Tao, J.~Kautz, and B.~Catanzaro (2018).
\newblock High-resolution image synthesis and semantic manipulation with conditional gans.
\newblock In {\em Proceedings of the IEEE conference on computer vision and pattern recognition}, pp.\  8798--8807.

\bibitem[\protect\citeauthoryear{Xu, Gotwalt, Hong, King, and Meeker}{Xu et~al.}{2020}]{xu2020applications}
Xu, L., C.~Gotwalt, Y.~Hong, C.~B. King, and W.~Q. Meeker (2020).
\newblock Applications of the fractional-random-weight bootstrap.
\newblock {\em The American Statistician\/}, 1--21.

\bibitem[\protect\citeauthoryear{Yu, Lu, and Stander}{Yu et~al.}{2003}]{yu2003quantile}
Yu, K., Z.~Lu, and J.~Stander (2003).
\newblock Quantile regression: applications and current research areas.
\newblock {\em Journal of the Royal Statistical Society: Series D (The Statistician)\/}~{\em 52\/}(3), 331--350.

\bibitem[\protect\citeauthoryear{Yu and Moyeed}{Yu and Moyeed}{2001}]{yu2001bayesian}
Yu, K. and R.~A. Moyeed (2001).
\newblock Bayesian quantile regression.
\newblock {\em Statistics \& Probability Letters\/}~{\em 54\/}(4), 437--447.

\end{thebibliography}
\endgroup

%

\newpage
\appendix
{

\centering{\Large{\bf Supplementary Materials for ``Generative Multi-purpose Sampler for Weighted M-esimtation''}}

}


\section{Theoretical Results}\label{sec:block}

\subsection{Theoretical Justification of Subgroup Bootstrap}\label{sec:theory_subgroup}
While  subgrouping was shown empirically to  approximate the target bootstrap distribution well, its theoretical  consistency is not immediately apparent. 
Here we employ the theoretical tools described in \cite{cheng2010bootstrap}
to examine theoretical aspects of subgrouping bootstrap procedures for the general M-estimation.
Let $Y_1,Y_2,\dots$ be a sequence of  iid random variables (with their observed values $y_1, y_2, \ldots$) from the probability distribution $ \mathbb{P}_0$ with the true parameter $\theta_0$, and the resulting expectation is denoted by $\mathbb{E}_0$. The probability distribution of $\bw$ and its expectation are denoted by $\mathbb{P}_\bw$ and $\mathbb{E}_\bw$,  respectively.  The empirical measure of the observations and the expectation with respect to it are denoted by $\widehat P_n$ and $\widehat{\mathbb{E}}_n$, respectively. We also define a weighted bootstrap empirical measure  $\widehat{ \mathbb{P}}_{\bw,n}=\sum_{i=1}^n w_i\delta_{y_i}/n$, where $\delta_{t}$ is a point measure at $t$, and the expectation with respect to it is denoted by $\widehat{ \mathbb{E}}_{\bw,n}$. 
We let $\ell'$ and $\ell''$ denote the first and second order derivatives of $\ell(\theta;Y)$ with respect to $\theta$, respectively. 
The original M-estimator, which corresponds to the solution of \eqref{eq:loss} with $\bw = {\boldsymbol 1}$, is denoted by $\hat \theta$. The big ``$O_{\mathbb P_0}$'' and small ``$o_{\mathbb P_0}$'' are based on probability distribution ${\mathbb{P}_0}$. We then consider some regularity conditions below: \\

\noindent{\bf (A1)} There exists  $\epsilon>0$ such that
$
\mathbb{E}_0 [\ell'(\theta)-\ell'(\theta_0)] = \mathbb{E}_0 [\ell''(\theta_0)(\theta-\theta_0)] + O(\norm{\theta-\theta_0}^2)
$, 
if $\norm{\theta-\theta_0}<\epsilon$ for large enough $n$.\\
\noindent{\bf (A2)} Suppose that 
$\mathbb{E}_0[\ell'(\theta_0)]=0$, $\widehat{\mathbb{E}}_n[\ell'(\hat\theta)]=o_{\mathbb P_0}(\sqrt{n})$, and $ \widehat{\mathbb{E}}_{\bw,n} [\ell'(\hat\theta_\bw)]=o_{\mathbb P_0}(\sqrt{n})$.
Also, assume that the optimizers $\hat\theta$ and $\hat\theta_\bw$ are unique over $\mathcal{W}$.  \\
\noindent{\bf (A3)} There exists $\delta>0$ such that  
$\widehat{\mathbb{E}}_n[\ell'(\theta)-\ell'(\theta_0)] - \mathbb{E}_0[\ell'(\theta)-\ell'(\theta_0)]  = O_{\mathbb P_0}(\sqrt{n}\norm{\theta-\theta_0})$, if $\norm{\theta-\theta_0}<\delta$.   \\
\noindent{\bf (A4)} The variance of  $\ell'(\theta_0)$ and $\mathbb{E}_0[\ell''(\theta_0)]$ are both non-singular.\\

These or similar mild regularity conditions  are often required for showing asymptotic consistency of M-estimators as in \citep{kosorok2008m, geer2000empirical}. Condition ({\bf A1}) assures that the derivative of the loss is smooth enough to be linearly approximated around  a local region at the true parameter. Condition ({\bf A2}) assumes  the uniqueness of the minimizers, and is general enough to deal with the case that the estimator and its bootstrap version are not exact minimzers, but ``nearly-minimizing'' the target losses. Condition ({\bf A3}) is called the {\it stochastic equi-continuity} \citep{cheng2010bootstrap} and   guarantees that the discrepancy between the empirical and true  derivatives of the loss around the true parameter is of order $\sqrt{n}$. Condition ({\bf A4}) ensures that the considered estimator asymptotically attains a non-singular variance.

Under the same regularity conditions as described in Section \ref{sec:Blockbootstrap},
\cite{cheng2010bootstrap} imposed the following set of additional  conditions on the weight distribution in order to guarantee the bootstrap consistency:\\

\noindent W1. The distribution of the weight vector $\bw=(w_1,\dots,w_n)$ is exchangeable for all $n=1,2,\dots$.  \\
\noindent W2. $\forall \ i$, $w_i\geq 0$, and $\sum_{i=1}^n w_i=n$.   \\
\noindent W3. $\exists \ C<\infty$, such that
$\limsup_{n\to\infty}\norm{w_1}_{2,1}\leq C,$ 
where $\norm{w_1}_{2,1} = \int^{\infty}_0 \sqrt{P_\bw(w_1\geq u)}du$.  \\
\noindent W4. $\lim_{\lambda\to \infty}\limsup_{n\to\infty}\sup_{t\geq \lambda}t^2P_\bw(w_1 > t)=0$.  \\
\noindent W5. $\sum_{i=1}^n(w_i-1)^2/n\overset{p}\to c^2$ with respect to $\mathbb{P}_\bw$ for some constant $c>0$.  \\

\begin{theorem}{\citep{cheng2010bootstrap}.} \label{theo:subgroup}
Assume that ({\bf A1}) -- ({\bf A4}) hold, and the subgroups are randomly assigned. Consider a random  weight bootstrap with $\bw \sim \mathbb{P}_{\bw}$ that satisfy W1--W5. Then, the resulting subgroup bootstrap is consistent; i.e.,
\begin{eqnarray*}
\sup_{x\in\mathbb{R}^p}\big\vert \mathbb P_{\bw\mid D_n}(\sqrt{n}(\hat\theta_\bw - \hat\theta) \leq x)-\mathbb{P}_{0}(\sqrt{n}(\hat\theta - \theta_0) \leq x) \big\vert \to 0,      
\end{eqnarray*}
in $\mathbb{P}_0$-probability as $n$ tends to $\infty$.
\end{theorem} 
Thus, it is sufficient to show that the subgroup bootstrap satisfies W1--W5. First, condition  W2 is trvially true. Since the marginal distribution of $w_1$ follows $\text{Beta}(1,S-1)$, conditions W3 and W4 are satisfied. For W5, because $\mathbb{E}\{(w_1-1)^2\} = 1- 2/(S+1)$, the subgroup weights satisfy W5.  

Now it is sufficient to show that the exchangeability holds as in W1 to show the consistency of the bootstrap. However, the members of each subgroup is fixed in advance, which  breaks the exchangeability condition among the bootstrap weights. Instead, we show that the subgroup bootstrap with fixed subgroups is consistent with  an exchangeable subgroup bootstrap with a random subgrouping.

As an opponent of the proposed fixed subgrouped weight $\bw_{S}=\{w_{S,1},\dots,w_{S,n}\}$ with deterministic subgroup indexes $\{I_1,\dots,I_S\}$, we first consider a fully randomized subgroup bootstrap weight $\widetilde \bw_{S}=\{\widetilde w_{S,1},\dots,\widetilde w_{S,1}\}$ that assumes the subgroups are also randomly assigned for every bootstrap evaluation. Like the subgroup weights, its weights are also generated from $\{\widetilde \alpha_{S,1},\dots,\widetilde \alpha_{S,S}\}\sim S\times \text{Dirichlet}(S,\mathbbm{1}_S)$. As a result, it is trivial that the distribution of $\widetilde \bw_{S}$ is exchangeable, as well as satisfying W2--W5. Then, it follows that for any bounded and continuous function $f$, 
\begin{eqnarray*} \nonumber
&&Var\left(\frac{1}{\sqrt{n}}\sum_{i=1}^n(w_{S,i}-1)f(y_i)\:\bigg\vert \:\by\right)=\mathbb{E}\left(\left\{\frac{1}{\sqrt{n}}\sum_{i=1}^n(w_{S,i}-\widetilde w_{S,i}+\widetilde w_{S,i}-1)f(y_i)\right\}^2\:\bigg\vert \:\by\right)\\ \nonumber
&=& \mathbb{E}\left(\left\{\frac{1}{\sqrt{n}}\sum_{i=1}^n(w_{S,i}-\widetilde w_{S,i})f(y_i)\right\}^2\:\bigg\vert \:\by\right)\ \ \ \cdot \:\cdot \:\cdot \:\cdot \:\cdot\:\cdot \:\cdot \:\cdot \:\cdot \:\cdot \: \cdot \:\cdot \:\cdot \:\cdot \:\cdot \ \ \ \text{(A1)}\\ \nonumber &+&2\mathbb{E}\left(\left\{\frac{1}{\sqrt{n}}\sum_{i=1}^n(w_{S,i}-\widetilde w_{S,i})f(y_i)\right\}\left\{\frac{1}{\sqrt{n}}\sum_{i=1}^n(\widetilde w_{S,i}-1)f(y_i)\right\}\:\bigg\vert \:\by\right)  \ \ \ \cdot \:\cdot \:\cdot \:\cdot \:\cdot  \ \ \ \text{(A2)}\\ \nonumber
&+&Var\left(\frac{1}{\sqrt{n}}\sum_{i=1}^n(\widetilde w_{S,i} - 1)f(y_i)\:\bigg\vert \:\by\right)\ \ \ \cdot \:\cdot \:\cdot \:\cdot \:\cdot\:\cdot \:\cdot \:\cdot \:\cdot \:\cdot \: \cdot \:\cdot \:\cdot \:\cdot \:\cdot \ \ \ \text{(A3)}. \nonumber
\end{eqnarray*}
Because the fully randomized subgroup bootstrap is consistent, the corresponding variance part (A3) in the above equation should be non-zero. As a result, it is sufficient to show that the other term (A1) + (A2) converges to zero in probability with respect to  the probability measure $\mathbb{P}_\by$. 

After a simple arithmetic, it follows that
\begin{eqnarray*} \nonumber
&&\text{(A1) + (A2)} \\ \nonumber &=&\frac{1}{n}\mathbb{E}\bigg(\bigg\{\sum_{i=1}^nw_{S,i}f(y_i)\bigg\}^2-\bigg\{\sum_{i=1}^n\widetilde w_{S,i}f(y_i)\bigg\}^2\:\Big\vert \:\by\bigg)\\ \nonumber
&=&\frac{1}{n}\mathbb{E}\bigg(\bigg\{\sum_{s=1}^S\alpha_{S,s}\sum_{i\in I_s}f(y_i)\bigg\}^2-\bigg\{\sum_{s=1}^S\widetilde\alpha_{S,s}\sum_{i\in \widetilde I_s}f(y_i)\bigg\}^2\:\Big\vert \:\by\bigg).
\end{eqnarray*}
Because $Var(\alpha_{S,1}) = \frac{S-1}{(S+1)}$ and $Cov(\alpha_{S,1},\alpha_{S,2}) = -\frac{1}{(S+1)}$, the above equation follows that 
\begin{eqnarray*} \nonumber
&&\text{(A1)$+$(A2)}\\ \nonumber
&=&\frac{S-1}{n(S+1)} \sum_{s=1}^S \Bigg[  \bigg\{  \sum_{i\in I_s}f(y_i)\bigg\}^2 - \mathbb{E}_{\widetilde I}\bigg(\bigg\{  \sum_{i\in \widetilde I_s}f(y_i)\bigg\}^2\:\Big\vert \:\by \bigg) \Bigg]  \ \ \ \cdot \:\cdot \:\cdot \:\cdot \:\cdot  \ \ \ \text{(B1)}  \\ \nonumber
&+&\frac{1}{n(S+1)} \sum_{s\neq k} \Bigg[  \bigg\{  \sum_{i\in I_s}f(y_i)\bigg\}\bigg\{  \sum_{i\in I_k}f(y_i)\bigg\} - \mathbb{E}_{\widetilde I}\bigg(\bigg\{  \sum_{i\in \widetilde I_s}f(y_i)\bigg\}\bigg\{  \sum_{i\in \widetilde I_k}f(y_i)\bigg\}\:\Big\vert \:\by \bigg) \Bigg].  
\end{eqnarray*}
We show that the variance of (B1) converges to zero as $n$ grows. Then, by using similar steps, we can show that the rest part converges to zero as well. Because $\sum_{s=1}^S\sum_{i\in I_s}f(y_i)^2=\sum_{i=1}^nf(y_i)^2$ for any exclusive subgroups $\{I_1,\dots,I_S\}$, it follows that
\begin{eqnarray*}
&&\sum_{s=1}^S\big\{\sum_{i\in I_s}f(y_i)\big\}^2 - \sum_{s=1}^S\mathbb{E}_{\widetilde I}\Big(\big\{  \sum_{i\in \widetilde I_s}f(y_i)\big\}^2\:\big\vert \:\by \Big) \\
&=&\sum_{s=1}^S\sum_{\substack{i,l\in  I_s\\i\neq l} }f(y_i)f(y_l)- \sum_{s=1}^S\mathbb{E}_{\widetilde I}\Big(  \sum_{\substack{i,l\in \widetilde I_s\\i\neq l} } f(y_i)f(y_l)\:\big\vert \:\by \Big).
\end{eqnarray*}
Therefore, it follows that 
\begin{eqnarray*}\nonumber
Var\{\text{(B1)}\} 
&=&Var\bigg[\frac{S-1}{n(S+1)} \sum_{s=1}^S \bigg\{\sum_{\substack{i,l\in  I_s\\i\neq l} }f(y_i)f(y_l)- \mathbb{E}_{\widetilde I}\bigg(  \sum_{\substack{i,l\in \widetilde I_s\\i\neq l} }f(y_i)f(y_l)\:\Big\vert \:\by \bigg)\bigg\}\bigg]\\ \nonumber
&\leq & 2Var\bigg[\frac{S-1}{n(S+1)} \sum_{s=1}^S \sum_{\substack{i,l\in  I_s\\i\neq l} }f(y_i)f(y_l)\bigg]\\ \nonumber
&=&\frac{2(S-1)^2}{n^2(S+1)^2}\sum_{s=1}^SVar\bigg[ \sum_{\substack{i,l\in  I_s\\i\neq l} }f(y_i)f(y_l)\bigg]\\ \nonumber
&=&\frac{2S(S-1)^2}{n^2(S+1)^2}\Bigg[  4(n/S)(n/S-1)Var\Big\{f(y_1)f(y_2)\Big\}\\ \nonumber
&+&8(n/S)(n/S-1)(n/S-2)Cov\Big\{f(y_1)f(y_2),f(y_1)f(y_3)\Big\} \Bigg]
\end{eqnarray*}
Then, the variance term is $O(S^{-1})$ and the covariance term is at a rate of $O(n/S^2)$. Therefore, the variance of (B1) converges to zero when $S\succ n^{1/2}$. \qed  

\section{Details of Double Bootstrap Procedures}\label{sec:details_double}
We first introduce notation here.  Let $F_0$  and $\widehat F_n$ denote the true distribution function and the empirical distribution of the observed data set, respectively. The bootstrapped version, which is resulted from random sampling the observations with replacement, of $\widehat F_n$ is denoted by $\widehat F_n^*$. In the same sense, the distribution of double bootstrapped observations, which is a bootstrapped version of $\widehat F_n^*$, is denoted by $\widehat F_n^{**}$. We denote the single-bootstrap and double bootstrap estimators resulted from $\widehat F_n^*$ and $\widehat F_n^{**}$ by $\hat\theta^*$ and $\hat\theta^{**}$, respectively. We let the expectation operators $\mathbb{E}_0$ and $\widehat{\mathbb{E}}_n$ be with respect to $F_0$ and $\widehat F_n$, respectively.

It is well-known that percentile (or bootstrap-t) CI via a single bootstrap procedure is not calibrated well in a sense that the resulting bootstrapped coverage is not matched to the nominal coverage, and the CI based on these procedures are unnecessarily wide  \citep{efron1994introduction}. For a one-sided CI with $95\%$ nominal coverage, the basic idea of these bootstrap is on the following approximation:
\begin{eqnarray*}
\mathbb{P}(T^* > t^*_{\alpha}\mid \widehat F_n)\approx \mathbb{P}_0(T > t_{\alpha}\mid  F_0) = 1-\alpha,   
\end{eqnarray*}
where {\it i)} $T=\hat\theta-\theta_0$ for the percentile procedure; {\it ii)} $T = (\hat\theta-\theta_0)/s$, where $s$ is  the standard error (when unknown, we set $s=1$), for the studentized procedure, and $t_{\alpha}$ is the  $\alpha$ quantile of the distribution of $T$. Also, $T^*$ and $t^*_\alpha$ are bootstrapped versions of $T$ and $t_\alpha$, respectively. Despite their simple and bootstrap-like intuition, the problem is that the bootstrapped probability $\mathbb{P}(T^* > t^*_{\alpha}\mid \widehat F_n)$ can be significantly deviated from the target coverage $1-\alpha$ in finite samples.  

To relieve this problem, \cite{hall1988bootstrap} considered a double bootstrap to calibrate the coverage error for the percentile procedure. This correction  searches for a valid quantile level $\widehat\alpha$ in a way that the resulting bootstrap coverage probability approximates $1-\alpha$; i.e., $\mathbb{P}(T^*>t^*_{\widehat\alpha}\mid\widehat F_n)\approx1-\alpha$. We can approximate such $\widehat\alpha$ by using a double bootstrap, and a bootstrap version of $\mathbb{P}(T^*>t^*_{\alpha}\mid\widehat F_n )$ can be evaluated via a double-bootstrapped probability $\mathbb{P}(T^{**}>t^{**}_{\alpha}\mid\widehat F_n^* )$, where $T^{**}$ and $t^{**}_\alpha$ are a double-bootstrapped counterpart of $T^*$ and $t^*_\alpha$. The solution $\widehat\alpha$ can be evaluated by a Monte Carlo approximation as the following steps:\\

\noindent 1. Evaluate $T_{b}^* = \hat\theta_b^* - \hat\theta$ and $T_{bc}^{**} = \hat\theta_{bc}^{**} - \hat\theta_b^*$ 
for $b=1,\dots,B$ and $c=1,\dots,C$.

\noindent 2. Construct $u_b^* = \frac{1}{C}\sum_{c=1}^C \mathbf{1}(T_{bc}^{**} > T_{b}^{*})$ for $b=1,\dots,B$.

\noindent 3. Set $\widehat\alpha = u^*_{(B(1-\alpha))}$, where $u^*_{(h)}$ is the $h$-th smallest ordered value of $\{u^*_1,\dots,u^*_B\}$.\\

\noindent Then, the calibrated  CI can be constructed by $(\infty, \hat\theta + t^*_{\widehat\alpha})$ for the percentile procedure. 

The other approach of double bootstraps is to estimate the standard error of $\hat\theta_b^* - \hat\theta$ for the studentized procedure \citep{hall1988theoretical}.  The explicit form of the bootstrap standard error $\hat s_b^*$ of $\hat\theta_b^*$ is frequently unknown, and the second level bootstrap of the $b$-th bootstrap data set can be used to evaluate the standard deviation of $\hat\theta_b^*$; i.e., $\hat s_b^*\approx\sqrt{\sum_{c=1}^C (T_{bc}^{**} - \bar T_b^*)^2 /(C-1)}$, where $\bar T_b^* = \sum_{c=1}^C T_{bc}^{**}/C$.  Then, the resulting one-sided CI with $95\%$ level is $(\infty, \hat\theta +  \tilde t^*_{95\%} \hat s )$, where $\tilde t_{\beta}^*$ is the $\beta$-quantile of $\{(\hat\theta_{b}^* - \hat\theta)/\hat s_b^*\}_{b=1,\dots,B}$, and $\hat s$ is the estimated standard error of $\hat \theta$ from the single bootstrap distribution. 
 By following a similar way, one can construct a two-sided confidence interval by changing lower and upper levels of quantile values. 

\section{Algorithm for CV generator training}

\begin{algorithm}[h!]
\footnotesize
\caption{\footnotesize  A subgroup weight distribution for (bootstrapped) $K$-fold CV.}\label{alg:CV}
\begin{algorithmic}
\STATE \textbf{Presetting:} 
\STATE $\bullet$ Set a candidate set of the tuning parameter $\{\lambda_1,\dots,\lambda_L\}$.
\STATE $\bullet$ Randomly subgroup the data points into $S$ blocks $\{I_1,\dots,I_S\}$ as demonstrated in  Section \ref{sec:Blockbootstrap}, and fix them. For simplicity, assume $\lfloor S/K  \rfloor$ to be an integer. Subgroup $\{I_1,\dots,I_S\}$ into $K$ folds, say $\{I^*_1,\dots,I^*_K\}$, where $I_{k}^*=\{I_{(k-1)S/K+1},\dots, I_{kS/K}\}$ for $k=1,\dots,K$ (each fold contains $S/K$ blocks). \\

\STATE \textbf{Sampling:} 
\STATE $\bullet$ Randomly select one $k'$ from $\{1,\dots,K\}$, and set $w_i=0$ for $i\in I^*_{k'}$, and the other weights are set to be 
\begin{equation*}
    \{w_i\}_{i\not\in I^*_{k'}}\begin{cases}
        =1 \text{ for an only CV,}\\
        \sim (S-S/K)\times\text{Dirichlet}(S-S/K,\mathbbm{1}_{S-S/K}) \\
        \text{ or } \text{Multinomial} (S-S/K,\mathbbm{1}_{S-S/K}/(S-S/K)) \text{ for a bootstrapped CV.}
\end{cases}
\end{equation*}
\end{algorithmic}
\vspace{-0.3cm}
\end{algorithm}

\subsection{Specifications of Computing System}
For the GMS and GBS applications, we used a GPU computing based on $2\times$\texttt{RTX2080ti} with 11GB RAM (a parallel GPU computing was not employed, but only a single GPU was used for each setting). For the conventional procedures, the computations were run under a workstation with a CPU of  
\texttt{Threadripper 2990WX} 64 threads with 128GB RAM.

\end{document}